\documentclass[11pt]{article}

\usepackage{amsmath,amssymb,amsthm}
\usepackage[shortlabels]{enumitem}
\usepackage{tikz}  
\usepackage[margin=1.1in]{geometry}
\usepackage{hyperref}
\usepackage{xcolor}
\usepackage{new_macro}
\usepackage{booktabs}        
\usepackage{stmaryrd}       
\newtheorem{theorem}{Theorem}[section]
\newtheorem{lemma}[theorem]{Lemma}
\newtheorem{hypothesis}[theorem]{Hypothesis}
\newtheorem{observation}[theorem]{Observation}
\newtheorem{problem}[theorem]{Problem}
\newtheorem{conjecture}[theorem]{Conjecture}
\newtheorem{fact}[theorem]{Fact}
\newtheorem{proposition}[theorem]{Proposition}
\newtheorem{corollary}[theorem]{Corollary}
\newtheorem{claim}[theorem]{Claim}
\theoremstyle{definition}
\newtheorem{definition}[theorem]{Definition}
\theoremstyle{remark}
\newtheorem{remark}[theorem]{Remark}
\numberwithin{equation}{section}
\newcommand{\C}{\mathbb{C}}
\newcommand{\cP}{\mathcal{P}}
\newcommand{\cS}{\mathcal{S}}
\newcommand{\bS}{\bar{\mathcal{S}}}
\newcommand{\bP}{\bar{\mathcal{P}}}
\newcommand{\wtS}{\mathrm{wt}_S}
\newcommand{\wtH}{\mathrm{wt}_H}
\newcommand{\Fix}{\mathrm{Fix}}
\newcommand{\spn}{\mathrm{span}}
\newcommand{\id}{\mathrm{id}}
\newcommand{\ket}[1]{|#1\rangle}
\newcommand{\bra}[1]{\langle #1|}
\newcommand{\braket}[2]{\langle #1|#2\rangle}
\newcommand{\G}{\mathrm{G}}

\newcommand{\Stab}{\mathcal{S}}

\title{Hardness of Approximating Quantum Code Distance Beyond $\sqrt{N}$}
\author{Upendra Kapshikar \\ University of Ottawa}
\date{}
\begin{document}
\maketitle

\begin{abstract}
    We study the computational complexity of approximating the minimum distance of a stabilizer code.
Classically, the minimum distance problem is NP-hard to approximate even to within an additive gap linear in the block length. For quantum codes, however, the reductions of Kapshikar and Kundu and of Grigorescu, Jha and Samperton reach only an additive $O(\sqrt{N})$ gap in the
block length $N$. 
We close this gap and, in a separate direction, initiate the fine-grained study of the problem.

We show that no randomized algorithm approximates the distance to within an additive $\alpha N$, for a constant $\alpha > 0$, unless $\mathsf{NP} \subseteq \mathsf{coRP}$. In the fine-grained setting, for every $\varepsilon > 0$, no randomized $2^{(1-\varepsilon)\kappa}\,\mathrm{poly}(N)$-time algorithm computes the distance of a code with $\kappa$ logical qubits unless SETH falls; and no randomized $2^{o(N)}$-time algorithm approximates the distance to within a linear additive gap, on instances where block length, number of logical qubits, and gap are simultaneously linear, unless
non-uniform Gap-ETH falls.
\end{abstract}
\newpage 
\tableofcontents
\newpage 

\section{Introduction}
\label{sec:intro}

A classical linear code (over $\mathbb{F}_2$) of block length $m$ and rank $n$ is an $n$-dimensional subspace $C\subseteq\F_2^m$, whose elements are called codewords.
It may be presented by a generator matrix, whose rows span $C$, or by a parity check matrix $H\in\F_2^{(m-n)\times m}$, so that $C=\{x:Hx=0\}$.
The minimum distance of $C$, written $\dist(C)$, is the smallest Hamming weight of a non-zero codeword, and it determines how many errors the code can detect and correct.

Two computational problems about linear codes have received most of the attention.
The nearest codeword problem, denoted $\NCP$, asks for the minimum Hamming distance between a given target vector and the code.
The minimum distance problem, denoted $\MDP$, asks for $\dist(C)$ itself.
The two look similar, but the requirement that the codeword be non-zero makes $\MDP$ a characteristically different problem from $\NCP$.
Berlekamp, McEliece and van Tilborg \cite{BMvT78} showed in 1978 that $\NCP$ is \NP-complete, and the analogous result for $\MDP$ came almost two decades later, with the work of Vardy \cite{Vardy} in 1997.
Dumer, Micciancio and Sudan \cite{DMS03} then showed that $\MDP$ remains \NP-hard under randomized reductions even in approximate form, to within any constant multiplicative factor and also to within an additive $\tau m$ for a constant $\tau>0$.
These results close off the most obvious route to codes of large distance.
If any of the problems was easy, one could sample parity check matrices at random and test whether the resulting code is good enough.

Quantum error-correcting codes have an analogous parameter.
The most commonly studied class of quantum codes is the stabilizer codes of Gottesman \cite{Dan_thesis}, which can be considered to be an analogue of linear codes.
A stabilizer group on $N$ qubits is an abelian subgroup $\Stab$ of the Pauli group that does not contain $-I$, and the code $Q$ it defines is the joint $+1$ eigenspace of its elements.
If $\Stab$ has $r$ independent generators, then $Q$ has dimension $2^{\kappa}$ for $\kappa=N-r$, and we call it an $\dbrack{N,\kappa}$ code with $N$ physical and $\kappa$ logical qubits.
The errors we consider are Pauli operators, and the weight of a Pauli is the number of qubits on which it acts non-trivially.
By the Knill--Laflamme conditions \cite{Laflamme}, a Pauli error is undetectable exactly when it commutes with every element of $\Stab$ but is not itself, up to a phase, an element of $\Stab$; such an error acts on the code space as a non-trivial logical operator.
The distance of $Q$, written $\qdist(Q)$, is the minimum weight of such an error: a Pauli lying in the normalizer of $\Stab$ but not in $\Stab$ itself.

The quantum minimum distance problem, denoted $\QMinDist$, asks for $\qdist(Q)$ given a set of generators of $\Stab$.
This resembles $\MDP$, but the two differ in what the minimum is taken over.
Classically it is taken over $C\setminus\{0\}$, a subspace with its zero vector removed.
Here it is taken over the normalizer of $\Stab$ with $\Stab$ itself removed, a group with a subgroup removed.
The elements of $\Stab$ are undetectable errors that act trivially on the code space, so they are never counted toward the distance, and a code is called degenerate if some of them have weight below $\qdist(Q)$.
Degeneracy has no classical analogue, and it is the point at which the correspondence between the two problems breaks.

The complexity of $\QMinDist$ has received some attention recently\cite{KK23,GJS25}.
Kapshikar and Kundu \cite{KK23} showed that it is \NP-complete, and that approximating $\qdist(Q)$, multiplicatively or additively, is \NP-hard under randomized reductions.
Their reduction goes through the codeword stabilized (\CWS) framework of Cross, Smith, Smolin and Zeng \cite{CSSZ09}, which builds a quantum code out of a graph and a classical code.
Grigorescu, Jha and Samperton \cite{GJS25} gave an alternative proof of these results using hypergraph product codes \cite{TZ14} instead, which gives a direct reduction to \CSS codes, along with a number of results on the distance of graph states.
Interestingly, as pointed out by \cite{GJS25}, both results, although using very different methods, manage to prove results in certain parameter region and their methods fail to go beyond this region. 
The classical result of \cite{DMS03} rules out approximation to within an additive $\tau m$, a constant fraction of the block length, but both quantum reductions above lose a square root and reach only $O(\sqrt N)$.
Since $\qdist(Q)\le N$ always, an additive gap of $\alpha N$ for a constant $\alpha$ is the strongest scale at which the promise is non-vacuous, and a gap of $\sqrt N$ is a considerable stretch away from it.
\cite{GJS25} in particular, point this out as a potential sqaure-root barrier.  
Kapshikar and Kundu point this out as a drawback of their proof technique and ask whether it can be removed \cite[\S V]{KK23}.
Interestingly, \cite{GJS25} showed that no additive hardness proportional to the block length can hold for hypergraph product codes at all (unless polynomial hierarchy collapses), and hence suggested that a genuine square root barrier might be at work for \QMinDist.
We show that this barrier is not a property of the generic \QMinDist problem.
We show that there is a constant $\alpha>0$ for which approximating $\qdist(Q)$ to within an additive $\alpha N$ is \NP-hard under randomized reductions.

Both prior results, and the one just stated, are \NP-hardness statements.
They rule out polynomial-time algorithms, but say nothing beyond that, for example, about the possibility of subexponential or even moderately faster exponential algorithms.
Classically, the best known way to compute $\dist(C)$ on worst-case codes is essentially exhaustive search over the $2^n$ codewords, and Stephens-Davidowitz and Vaikuntanathan \cite{SV19} showed that this is optimal under \SETH: for every constant $\eps>0$ there is no algorithm running in time $2^{(1-\eps)n}$.
They also showed that approximating $\dist(C)$ to within a constant factor requires time $2^{\Omega(n)}$ under non-uniform Gap-\ETH.
Nothing of this kind was known for $\QMinDist$.
We give the first such lower bounds for the quantum problem, ruling out algorithms running in time $2^{(1-\eps)\kappa}$\footnote{in fact, we give a stronger lower bound than this, with an additional factor depending on $N$} under \SETH, and $2^{o(N)}$-time algorithms at a linear additive gap under non-uniform Gap-\ETH.

\subsection{Our results}
All of our reductions produce codeword stabilized (\CWS) codes, introduced by Cross, Smith, Smolin and Zeng \cite{CSSZ09}.
A \CWS code is built from two classical objects, a graph $G$ on $N$ vertices and a classical code $C\subseteq\F_2^N$, and when $C$ is linear the resulting code is a stabilizer code, whose stabilizer presentation is computable from $G$ and $C$ in polynomial time.
Our results are therefore stated in the standard stabilizer input format.

Throughout, $N$ denotes the number of physical qubits of a stabilizer code and $\kappa$ its number of logical qubits.
For a gap function $g$, we write $\GapAddQDist_g$ for the promise problem of distinguishing $\qdist(Q)\le t$ from $\qdist(Q)>t+g(N)$, and in particular, $\GapAddQDist_{\alpha N}$ denotes $g(N)=\alpha N$ for a constant $\alpha>0$.

Our first result places the additive gap at the scale of the block length.

\begin{namedthm}[linear additive gap]\label{thm:A}
There exist constants $\alpha_1,\alpha_2>0$ such that
\begin{enumerate}
\item[(a)] $\GapAddQDist_{\alpha_1N}\notin\BPP$ unless $\NP\subseteq\BPP$; and
\item[(b)] $\GapAddQDist_{\alpha_2N}\notin\coRP$ unless $\NP\subseteq\coRP$, which collapses $\NP=\RP=\coRP=\ZPP$.
\end{enumerate}
\end{namedthm}

Part (b) is stronger than the corresponding classical statement, and the reason is worth pointing out.
The reduction of \cite{DMS03} is faithful on \NO~instances and unfaithful on \YES, which yields $\NP=\RP$.
Ours is unfaithful in the opposite direction.
The \YES~conclusion $\qdist(\CWS(G,C))\le\dist(C)$ holds for every graph $G$, whether or not the random sampling of the reduction succeeded, so only \NO~instances can be corrupted.

The next two results are fine-grained lower bounds, similar to \cite{SV19}, who establish such lower bounds for classical linear codes.

\begin{namedthm}[\SETH~hardness]\label{thm:B}
Assume randomized \SETH.
For every constant $\eps>0$ there is no randomized algorithm that, given an $\dbrack{N,\kappa}$ stabilizer code $Q$ and an integer $t$, decides whether $\qdist(Q)\le t$ in time $2^{(1-\eps)\kappa}\poly(N)$.
\end{namedthm}

Our reduction preserves the rank exactly, so $\kappa = n$ and the classical exponent transfers without loss.
Theorem~\ref{thm:B} is, to our knowledge, the first lower bound of this kind for $\QMDP$: prior work \cite{KK23,GJS25} rules out only polynomial-time algorithms. The bound here excludes an exponential family in the parameter that governs the size of the code space, at the same exponent that \cite{SV19} obtains classically.
It is not necessarily tight, however, because of a genuine asymmetry relative to the classical case.
Codeword enumeration decides $\MDP$ in time $2^{n}\poly(m)$, which matches the lower bound of \cite{SV19}.
The corresponding quantum brute force enumerates $\bar{\mathcal{S}}^{\perp} \setminus \bar{\mathcal{S}}$, of size $2^{N+\kappa} - 2^{N-\kappa}$ by Definition~\ref{def:qdist}, and no algorithm exponential in $\kappa$ alone is known.
The search does not shrink with $\kappa$: already at $\kappa = 1$ one faces a minimum-weight problem over a space of dimension $N+1$.
Since $\kappa \le N$ always, Theorem~\ref{thm:B} constrains a regime well below the best known algorithm.
\begin{namedthm}[Gap-\ETH~hardness]\label{thm:C}
Assume non-uniform Gap-\ETH.
There are constants $\alpha>0$ and $c_\kappa>0$ such that no randomized $2^{o(N)}$-time algorithm solves $\GapAddQDist_{\alpha N}$ on stabilizer codes with $\kappa\ge c_\kappa N$.
\end{namedthm}

The hardness of \Cref{thm:C} therefore holds on instances that are simultaneously linear in the block length, the number of logical qubits, and the additive gap.
Keeping all three linear is what the classical side of this paper requires work for, and we return to it in \Cref{sec:overview}.

\paragraph{Removing the randomness.}
The reductions behind Theorems~\ref{thm:A}--\ref{thm:C} are randomized, and the randomness enters in exactly one place: the choice of the graph.
We remove it under a standard hypothesis.

\begin{namedhyp}[see \Cref{sec:hypD} for more details]
There is a language in $\E=\DTIME(2^{O(n)})$ that requires nondeterministic Boolean circuits of size $2^{\Omega(n)}$.\label{hyp:intro}
\end{namedhyp}

This is the assumption underlying \AM~derandomization \cite{KvM02,MV05}, the nondeterministic analogue of the Impagliazzo--Wigderson hypothesis that $\E$ requires exponential-size deterministic circuits \cite{IW97}.

\begin{namedthm}[deterministic forms]\label{thm:D}
Assume the circuit hypothesis given above.
Then the conclusions of Theorems~\ref{thm:A}--\ref{thm:C} hold with the reductions made deterministic, at the cost of becoming nonadaptive Turing rather than many-one.
In particular, there is a constant $\alpha>0$ such that $\GapAddQDist_{\alpha N}$ is \NP-hard under deterministic polynomial-time truth-table reductions, and, assuming additionally deterministic \SETH, no deterministic $2^{(1-\eps)\kappa}\poly(N)$-time algorithm decides $\QMinDist$.
\end{namedthm}

Our fine-grained bounds are parametrized by $\kappa$ rather than by $N$, and this is inherited from the classical input: the \SETH-hardness of \cite{SV19} is stated in the rank, and on their instances the block length is polynomially larger.
Hardness of $\MDP$ in the block length is their own conjecture \cite[\S1.2]{SV19}.
We show that if the conjecture holds, then such a result translates to the quantum problem.

\begin{conjecture}[{\cite[\S1.2]{SV19}}]\label{conj:intro-sv}
There is a constant $\beta>0$ such that no algorithm decides $\MDP$ on binary codes of block length $m$ in time $2^{\beta m}$.
\end{conjecture}

\begin{namedthm}\label{thm:E}
Assume the circuit hypothesis given above. 
If the above conjecture by \cite{SV19} holds, then there is an explicit constant $c>0$, depending only on $\beta$, such that no deterministic algorithm decides $\QMinDist$ on stabilizer codes of block length $N$ in time $2^{cN}$.
\end{namedthm}

Conjecture~\ref{conj:intro-sv} places no restriction on the rank or threshold of the instances, and this is deliberate.
Codeword enumeration decides rank-$\kappa$ instances in time $2^\kappa\poly(m)$, and weight enumeration decides instances of small threshold in comparable time, so any such restriction would cap the believable $\beta$ well below $1$.
The unrestricted form is a stronger hypothesis and precisely the question asked in \cite{SV19}.

\paragraph{\CSS~codes.}
\CSS~codes are arguably the more popular quantum codes, and every result above translates to them, by an argument given by \cite{KK23}, who use a reduction given by \cite{BTL10}.
We include it here once for completeness. 
The main tool is the doubling map of Bravyi, Terhal and Leemhuis \cite{BTL10}, which routes a stabilizer code through a Majorana fermion code.
It takes an $\dbrack{N,\kappa}$ stabilizer code of distance $d$ to a $\dbrack{4N,2\kappa}$ \CSS~code of distance exactly $2d$, in polynomial time.
Because the output distance is exactly $2d$, both sides of a promise move together and the gap versions survive intact.

\begin{namedcor}[\CSS~forms, informal]\label{cor:css-intro}
Theorems~\ref{thm:A}, \ref{thm:C} and \ref{thm:D} hold for \CSS~codes with each additive-gap constant halved, and all $2^{o(N)}$ and \NP-hardness forms unchanged.
\Cref{thm:B} holds with $2^{\left(1-\eps\right)\kappa}$ replaced by $2^{(\frac{1}{2}-\eps)\kappa}$, and \Cref{thm:E} with $c$ replaced by $\frac{c}{4}$.
\end{namedcor}
\begin{proof}
    Compose each reduction with the doubling map; the composition is still polynomial time.
For a promise with threshold $t$ and additive gap $g$, the image satisfies $\qdist\le2t$ on \YES~and $\qdist>2t+2g$ on \NO, so the additive gap is $2g$ against block length $N'=4N$: writing $g=\alpha N=\alpha N'/4$ gives $2g=(\alpha/2)N'$.
An algorithm running in time $2^{(1/2-\eps)\kappa'}\poly(N')$ decides the stabilizer instance in time $2^{(1-2\eps)\kappa}\poly(N)$, contradicting \Cref{thm:B} at parameter $2\eps$; and one running in time $2^{cN'}=2^{4cN}$ contradicts \Cref{thm:E} whenever $4c<\beta\eps_1$.
\end{proof}

\begin{table}[t]
\centering\small
\begin{tabular}{@{}llll@{}}
\toprule
Result & Hypothesis & Conclusion & Reduction \\
\midrule
\Cref{thm:A}(a) & No hypothesis needed &
  $\GapAddQDist_{\alpha_1N}\notin\BPP$ unless $\NP\subseteq\BPP$ & randomized \\
\Cref{thm:A}(b) & No hypothesis needed &
  $\GapAddQDist_{\alpha_2N}\notin\coRP$ unless $\NP=\RP=\coRP=\ZPP$ & randomized \\
\Cref{thm:B} & rand.\ \SETH &
  no $2^{(1-\eps)\kappa}\poly(N)$ randomized & Karp (rand.) \\
\Cref{thm:C} & non-unif.\ Gap-\ETH &
  no $2^{o(N)}$ at gap $\alpha N$, $\kappa=\Theta(N)$ & Karp (rand.) \\
\Cref{thm:D} & circuit hyp. $+\,\SETH$ &
  deterministic forms of Theorems~\ref{thm:A}--\ref{thm:C} & det.\ truth-table \\
\Cref{thm:E} & circuit hyp. $+$ Conj.~\ref{conj:intro-sv} &
  no $2^{cN}$ deterministic & det.\ truth-table \\
\bottomrule
\end{tabular}
\caption{Summary of results.}
\label{tab:results}
\end{table}

\subsection{Technical overview}
\label{sec:overview}

In this section, we sketch our proofs.
Every reduction builds a \CWS~code from the classical instance and a graph.
The graph is where the first layer of techincal choice shows up: it must be chosen in a way so that it has a large enough \emph{graph distance} $\Gdist$.
\cite{KK23} choose their graph by a structural condition, namely that the graph family must be 4-cycle free, and argue about each low-weight error separately; we do this in two steps, first we sample a random graph (which gives results of the form \emph{unless \NP $ \subseteq $ co\RP}) and then under circuit hypothesis we derandomize it (which then leads to results).
However, merely choosing a graph with a good graph distance is not enough for reductions, and in particular, \cite{KK23}-style reduction does not give quantum codes with large enough distance. For example, consider the construction of \CWS~codes given by \cite{KK23}. 
Even if one replaces 4-cycles free graphs, which have $\Gdist=\Theta(\sqrt{n})$, by a family with linear $\Gdist$, it does not give quantum codes linear distance. 
So, a random graph sampling with linear $\Gdist$ is only half of the solution, and perhaps the easier half.
Of course, the random choice is also what has to be removed to make the reductions deterministic, and we do it under the circuit hypothesis, using a hitting set in place of the sampling, which is a well-established tool to do so.

\subsubsection*{CWS codes}
Let $\sigma(E)=(a\mid b)\in\F_2^{2N}$ be the symplectic image of a Pauli $E=X(a)Z(b)$, and $\wt_S(a\mid b)=|\{i:(a_i,b_i)\neq(0,0)\}|$ be its (symplectic) weight.
A \CWS~code $\CWS(G,C)$ is built from a graph $G$ on $[N]$, with adjacency matrix $A_G$, and a linear code $C\subseteq\F_2^N$.
Its codewords are the states $Z(c)\ket{G}$ for $c\in C$, where $\ket{G}$ is the graph state of $G$.

The error correction condition runs through \emph{classicalization map} introduced in \cite{CSSZ09}.
An $X$-error on a vertex travels along the edges of $G$ and reappears as a $Z$-error on the neighbourhood, so a Pauli $E$ with $\sigma(E)=(a\mid b)$ induces the classical error
\[
  \Cl_G(E)\;:=\;b\oplus A_G a\;\in\;\F_2^N.
\]
The detection criterion for \CWS codes is then stated in terms of this induced error (see Preliminaries for more details).
Here we give brief description of how undetectable errors for a \CWS~code look like since they determine the minimum distance of the code. 
Formal definitions can be found in the main body of the paper. See for example, Lemma~\ref{lem:exact}.

An undetectable, logically non-trivial Pauli falls into exactly one of two classes.
First, the \emph{nonzero-syndrome} errors have $a\neq0$ and $\Cl_G(E)\in C\setminus\{0\}$.
And the second, \emph{zero-syndrome} errors have $a\neq0$ and $\Cl_G(E)=0$, that is $b=A_Ga$.
The distance is the minimum weight over these errors.
The zero-syndrome errors are governed by the graph alone.
The least weight of a nonzero $E$ with $a\neq0$ and $\Cl_G(E)=0$ is the graph state distance $\Gdist(G)$, which depends on $G$ and not on $C$.

\subsubsection*{The square-root barrier}

As mentioned by \cite{KK23} the \CSS~construction, the usual route from a classical code to a quantum one, is not necessarily directly available here.
It requires an orthogonal pair $C_1^\perp\subseteq C_2$, so a hard instance $C_1$ would have to be paired with a $C_2$ that is orthogonal to it and has distance no smaller.
It is not known whether classical hardness survives this restriction.
The \CWS framework places no condition on $C$, which is why both \cite{KK23} and we take this route.
\cite{GJS25} do manage to to give a direct \CSS~reduction, however they also show that it is impossible to go beyond the $\sqrt{N}$ approximation by their technique, unless the polynomial hierarchy collapses.
Thus, we choose to operate in the framework of \CWS~codes, with motivation similar to that of \cite{KK23}.
In a \CWS~code, for a badly chosen graph $G$, classicalization described above produces a low weight undetectable error, and then $\qdist(\CWS(G,C))\ll\dist(C)$.
A \CWS~reduction must therefore use graphs that do not destroy the distance of the code paired with them.

As said before, this barrier is easy to be misplaced.
A graph of minimum degree $\delta$ has $\Gdist(G)\le\delta+1$, and \cite{KK23} saturate this bound from a $4$-cycle-free graph family.
However, for a 4-cycle free family, the minimum degree can not exceed $\Theta(\sqrt{N})$, and hence the approach of \cite{KK23} also gets stuck at this $\sqrt{N}$ barrier, as highlighted by \cite{GJS25}.  
It is tempting to conclude that a denser graph, with $\Gdist(G)=\Omega(N)$, would remove the obstruction.
First of all, to our knowledge, efficient constructions of such graphs is not known.
Secondly, and more importantly, this replacement does not complete the picture.
A random graph already has $\Gdist(G)=\Omega(N)$ with overwhelming probability, and the zero-syndrome errors that $\Gdist(G)$ governs are the half of the problem that \cite{KK23} already control without much difficulty.
The barrier lies in the nonzero-syndrome errors, on which $\Gdist(G)$ says nothing, since $\Gdist$ concerns only errors whose classicalization is zero.

For the nonzero-syndrome errors, \cite{KK23} argue by weight enumeration.
First they reduce the classical \MDP to the case where codes have sparse codewords, that is, maximum weight of a code instance is upper bounded by $\sqrt{N}$.
Then they bound $\wt_H(\Cl_G(E))$ from below, one error at a time, and conclude that $\Cl_G(E)$ is too large to be a codeword of $C$.

A denser graph would raise the floor but leave the ceiling in place.
It appears that, thus, one obtains tight bounds exactly in the regime where \cite{KK23} use it, which is the $\sqrt{N}$ approximation region.
This is also why a random graph, even though it has a large $\Gdist(G)$, is useless on its own for the weight argument.
Hence, we bypass this sandwiching idea of \cite{KK23}; instead, we give a more localized criterion for \emph{compatibility}.
One way to look at our proof technique would be the following:
Among other things, one key ingredient of \cite{KK23}'s Karp reduction is that 4-cycle free graphs that are universally \emph{compatible} with any (sparse) classical code $C$.
We give a truth-table reduction (under circuit hypothesis), that is not universal. 
Instead, given a code $C$, we find a \emph{compatible} graph for $C$.

\subsubsection*{Compatibility}
\begin{definition}[informal; also see Def.~\ref{def:compatible}]
Fix a linear code $C\subseteq\F_2^N$ and an integer $w\ge1$.
A graph $G$ is \emph{$(C,w)$-compatible} if every Pauli $E$ with $\sigma(E)=(a\mid b)$, $a\neq0$, and $\wt_S(a\mid b)<w$ has $\Cl_G(E)\notin C\setminus\{0\}$.
It is \emph{$w$-diagonal} if $\Gdist(G)\ge w$.
\end{definition}

Compatibility handles the nonzero-syndrome errors and diagonality handles the zero-syndrome ones, so together they can be used lower-bound the distance.
The condition is derived from the staandard detection criterion for \CWS~codes, but differs from it in three ways.

First, only the nonzero-syndrome clause becomes a condition on the graph relative to the code.
The zero-syndrome clause is delegated to $w$-diagonality, a property of the graph alone.
This is the split we need: diagonality is one half, which as epected, is satisfied by a random graph followed by derandomization, and compatibility is the other half that $\Gdist(G)$ cannot reach.

Second, the pure-$Z$ errors, those with $a=0$, are excluded and handled separately.
For these the graph drops out of $\Cl_G(E)=b$, and the pure-$Z$ errors at codewords are the intended logical operators, the ones carrying $\dist(C)$ into $\qdist$.
Forbidding them would make the condition unsatisfiable on \YES~instances.
Their weight is controlled separately, by $\dist(C)$.

Third, the condition is truncated at the level $w$.
The reason for this truncation is too techincal for the proof overview, but gets clearer once the proof is setup.
In short, the truncation is chosen to balance between two things: the level is large enough so that we get large enough lower bound on $\Gdist$ but small enough so that the union bound that we get from probabilistic arguments does not blow up. 

\subsubsection*{Random graphs are compatible}

Compatibility and diagonality are established by a first-moment argument, resting on one fact about random graphs.

\begin{claim}[see Claim~\ref{claim:half_uniformity}]
For a uniform graph $G$ on $[N]$ vertices and any fixed $a\neq0$, the vector $A_Ga$ is uniformly distributed on $a^\perp$.
\end{claim}

The map $A\mapsto Aa$ is $\F_2$-linear on the space of symmetric zero-diagonal matrices.
Its image lies in $a^\perp$, since $x^\top Ax=0$ for such $A$, and the reverse inclusion follows by exhibiting, for each $y\in a^\perp$, a matrix sending $a$ to $y$.

\medskip \noindent 
\emph{Compatibility:} For a fixed error $E = (a \mid b)$ with $a \neq 0$, the syndrome 
$\Cl_G(E) = b \oplus A_G a$ is uniform over the coset $b \oplus a^\perp$ of size $2^{N-1}$, so
\[
\Pr_G\bigl[\Cl_G(E) \in C'\bigr] \leq 2^{k'-N+1}.
\]
A union bound over all errors of Pauli weight below $w^*$ gives the required guarantee provided $k' \leq \frac{N}{4}$ and 
$w^* \leq \varepsilon_1 N$ for a small enough constant $\varepsilon_1 > 0$. Both conditions are 
achieved by padding the classical instance to length $N = O(w + m)$, which preserves rank and 
distance. Diagonality follows by an analogous argument. At sufficiently small $\varepsilon_1$, 
both failure probabilities are $2^{-\Omega(N)}$.

\subsubsection*{From compatibility to a distance bracket}

Compatibility and diagonality together place $\qdist(\CWS(G,C'))$ between the classical distance and the associated level.

\begin{lemma}[informal; see \Cref{lem:threshold}]
Let $C'$ be a code of distance $\lambda$, let $t,g\ge0$, and set $w^*:=t+g+1$.
For every graph $G$, the code $Q_G:=\CWS(G,C')$ satisfies $\qdist(Q_G)\le\lambda$, witnessed by a pure-$Z$ Pauli.
If in addition $G$ is $(C',w^*)$-compatible and $w^*$-diagonal, then $\qdist(Q_G)=\lambda$ when $\lambda\le t$, and $\qdist(Q_G)>t+g$ when $\lambda>t+g$.
\end{lemma}

The lower bound is a three-way case split over the detection criterion.
A zero-syndrome error has weight at least $\Gdist(G)\ge w^*$ by diagonality.
A nonzero-syndrome error of weight below $w^*$ is ruled out by compatibility, so any error that remains has weight at least $w^*$.
A pure-$Z$ error at a codeword has weight at least $\lambda$.
Hence $\qdist(Q_G)\ge\min(w^*,\lambda)$, and the two regimes follow by comparing $\lambda$ to $t$ and to $t+g$.

The upper bound $\qdist(Q_G)\le\lambda$ holds for \emph{every} graph, and is witnessed by $Z(c)$ at a minimum-weight codeword $c$.
This makes the reduction one-sided, so that only \NO~instances can be corrupted by a bad graph.
That is what yields the $\coRP$ conclusion of \Cref{thm:A}(b)  and further makes a hitting set rather than a pseudorandom generator suffice below, which is critical for the truth-table selection rule of \Cref{thm:D}.

\subsubsection*{Fine-grained reductions}

The fine-grained results demand that no step of a series of compositions loses more than it can afford. 
For \Cref{thm:B}, the chain of reductions is
\[
  k\text{-SAT}
  \;\xrightarrow{\ \cite{SV19}\ }\;
  \NCP
  \;\xrightarrow{\ \cite{SV19}\ }\;
  \MDP
  \;\xrightarrow{\text{\ this work\ }}\;
  \text{stabilizer code},
\]
where the first two steps are the \SETH-hardness reductions of \cite{SV19}.

The reduction of \cite{SV19} is stated for algorithms running in time $2^{(1-\eps)n}$, where $n$ is the rank of the classical code, with no polynomial factor attached.
Our reduction preserves the rank exactly, so the quantum instance has $\kappa=n$, and the exponent transfers without loss.
What does not transfer is the polynomial factor.
\Cref{thm:B} rules out algorithms running in time $2^{(1-\eps)\kappa}\poly(N)$, and since the quantum block length $N$ is polynomial in $n$, composing with our reduction yields a classical algorithm running in time $2^{(1-\eps)n}\poly(m')$, which the statement of \cite{SV19} does not rule out.
We therefore need the strengthened form, as stated in \Cref{thm:strong}, which rules out $2^{(1-\eps)n'}\poly(m')$.

The remaining subtlety is the deterministic derandomization step in \cite{SV19}, which produces 
$2^{3n/4}$ classical instances and costs $2^{3(1+\varepsilon')n/4}$ to build them. A short 
calculation shows that the combined exponent of the full reduction stays strictly below $1$ for 
a suitable choice of $\varepsilon'$, so the \SETH~lower bound is preserved throughout the 
composition.

For \Cref{thm:C} every parameter must stay linear, and Observation~\ref{prop:linprofile} verifies this through the reduction of \cite{SV19}.
The gap reduction of \cite{SV19} passes through a \emph{locally dense code}, a gadget with exponentially many codewords at the minimum distance from a target, which is what allows a nearest-codeword instance to be embedded into a minimum-distance instance.
Their gap gadget is built from the kissing-number codes of Ashikhmin, Barg and Vl\u{a}du\c{t} \cite{ABV01}, and it has rank, block length, and radius all $\Theta(n)$.
Tracking these through the reduction gives a $\gamma_2$-\MDP instance with rank, block length, and threshold all $\Theta(n)$, and converting the constant multiplicative gap to an additive one gives a gap $\Theta(m)$.
Padding then keeps all three of block length, rank, and additive gap linear in $N$ simultaneously, which is the content of \Cref{thm:C}.

One point in this chain is not stated in \cite{SV19} but needed here: their gadget radius is bounded only from above, whereas the matching lower bound is what forces the number of gadget copies to be constant, and hence keeps the construction linear (see Observation~\ref{prop:linprofile}).

\subsubsection*{Derandomization}
The randomness in our reductions enters in one place, the choice of $G$, and as mentioned before, we remove it under a circuit hypothesis.
Let $B$ be the set of \emph{bad} graphs for a given code and a given level, the graphs that are not $(C',w^*)$-compatible or not $w^*$-diagonal.
We note three facts about $B$.
\begin{enumerate}
    \item The set $B$ is recognized by a co-nondeterministic circuit of polynomial size.
    A bad graph is certified by a single witness $(a\mid b)\in\F_2^{2N}$: the verifier computes $A_Ga$, forms $\Cl_G(E)$, tests its weight, and tests membership in $C'\setminus\{0\}$ against a hardwired parity check, all in time $O(N^2)$.
    \item The set $B$ is sparse, with density $2^{-\Omega(N)}$ by the first-moment argument above, so its complement has density well above $\frac{1}{2}$.
    \item The reduction is one-sided, by the upper bound of the threshold lemma, so we need only \emph{hit} the good set, never estimate its measure.
\end{enumerate}
A hitting-set generator is an object that intersects every dense set that is recognizable by a co-nondeterministic circuit, and it follows from a weaker hardness assumption than a full pseudorandom generator.
Under the circuit hypothesis, such a generator (for example, that of \cite{MV05}) yields, in deterministic polynomial time, a list $G_1,\dots,G_r$ of $\poly(N)$ graphs containing at least one good graph.
The reduction queries the oracle on $(Q_{G_i},t)$ for each $i$ and answers \NO~if any query does.
On a \YES~instance every $Q_{G_i}$ has $\qdist\le\lambda\le t$ by the threshold lemma.
So every query answers \YES; on a \NO~instance the one good graph supplies a \NO~vote, and the remaining queries may violate the promise and be answered arbitrarily without harm.
This is why the reductions become truth-table rather than many-one, and this one-sidedness makes the arbitrary answers harmless.

\subsection{Discussion and open problems}
\label{sec:open}

We collect some of the questions this work leaves open.
Each is a place where the method stops for an identifiable reason, and in three cases, we can say what an answer would likely have to supply.

\paragraph{A many-one reduction.}
The deterministic reductions of \Cref{thm:D} are truth-table rather than many-one.
They query one candidate code per graph in the hitting set and take the conjunction of the answers.
A many-one version would have to combine the candidates into a single stabilizer code whose distance tracks the \emph{maximum} of the candidates' distances, so that one good graph forces the whole code onto the \NO~side.
Natural code combinations give the minimum instead.
We do not know whether a single code can be made to track the maximum, and this is the obstacle.

\paragraph{An unconditional derandomization.}
Both \Cref{thm:D} and \Cref{thm:E} need a compatible graph produced deterministically, and under the circuit hypothesis the hitting set supplies a short list containing one.
Without the hypothesis one would need an explicit family of compatible graphs at a linear level.
Searching for one by brute force costs $2^{c_1(\eps_1)N}$, which is already too slow for the polynomial-time reductions of \Cref{thm:D}; and for \Cref{thm:E}, where exponential time is available, it seems that the construction still costs more than the hardness it would preserve.
An explicit construction is therefore what an unconditional version requires.
\paragraph{Explicit compatible graphs, and a classical connection.}
For block graphs, $w^*$-diagonality is equivalent to a lower bound of $w^*$ on $\min\{\dist(C),\dist(C^\perp)\}$ for an associated rate-$\frac{1}{2}$ code, a decoupling proved in \cite[proof of Thm.~2]{GJS25}.
Explicit codes with $\min\{\dist(C),\dist(C^\perp)\}=\Omega(n)$ are exactly the object asked for in \cite[\S5.3]{GJS25}, so the diagonality half of an explicit construction is precisely their question.
Compatibility is the strictly stronger requirement.
Under the circuit hypothesis our list construction produces, in deterministic polynomial time, a $\poly(N)$-size list of graphs at least one of which is diagonal at a linear level, which answers the \emph{list} version of their question; the fully explicit version remains open.

\paragraph{Sparse instances.}
The graphs produced by our construction are dense, so the stabilizer generators $g(h) = \pm X(h)Z(A_G h)$ of the codes we output have weight $\Theta(N)$.
Our results, therefore, say nothing about $\QMDP$ restricted to LDPC codes, where the generators have bounded weight, and which is the regime of most of the practical interest in quantum codes.
Compatibility at a linear level appears to require density: the first-moment argument we use rests on $A_G a$ being uniform on $a^{\perp}$, and the overall idea seems to fail for sparse $A_G$.
Whether hardness under a linear additive gap holds for LDPC instances remains open, and an answer would require a source of graphs beyond those we consider.

\paragraph{The right fine-grained parameter.}
Theorem~\ref{thm:B} rules out $2^{(1-\varepsilon)\kappa}\poly(N)$ under \SETH, and Theorem~\ref{thm:C} rules out $2^{o(N)}$ but only under non-uniform Gap-ETH.
The obvious question is whether the two can be merged: is there a $2^{\Omega(N)}$ lower bound for $\QMDP$ under SETH?
Unlike the questions above, this is one where we cannot say what an answer must supply, because it may resolve in either direction, and the two directions correspond to the two ways the present bound could fail to be tight.
The first is algorithmic.
A $2^{O(\kappa)}\poly(N)$ time algorithm for $\QMDP$   would make Theorem~\ref{thm:B} tight up to the constant in the exponent, establishing $\kappa$ as the correct parameter for the quantum problem, similar to the classical one.
We are not aware of one, and the enumeration bound noted after Theorem~\ref{thm:B} suggests a bottleneck: the search does not shrink with $\kappa$.
The second approach is to obtain SETH-hardness on instances with $\kappa = \Theta(N)$, where a bound parametrized by $\kappa$ is automatically one parametrized by $N$.
Theorem~\ref{thm:C} already produces such instances, so the difficulty does not lie in our reduction, which preserves the rank exactly; it is that the SETH-hardness of \cite{SV19} is stated in the rank, and their instances have block length polynomially larger.
This is the same bottleneck that Theorem~\ref{thm:transfer} addresses conditionally, and it is a classical one.
Our reduction sends rank to logical qubits without loss and provides evidence that the analogy between the two parameters holds at the level of reductions; the fact that no $2^{O(\kappa)}$ algorithm is known suggests it may fail at the level of algorithms.
Which one of the two holds, to us, the more interesting form of the question.
\paragraph{A rank-preserving reduction to \CSS~codes.}
The route to \CSS~codes in Corollary~\ref{cor:css-intro} doubles the number of logical qubits, and this is what costs a factor in the exponent of the fine-grained bounds.
The hypergraph product reduction of \cite{GJS25} reaches \CSS~codes directly, without the doubling map, but cannot carry a linear additive gap.
A reduction that keeps the linear gap and does not double the rank would remove the loss, and we regard this as the natural next step.

\section{Preliminaries}

\subsection*{Notation}
\begin{itemize}
\item For $x, y \in \F_2^N$, $x \cdot y := \bigoplus_i x_i y_i$ denotes the standard inner product over $\F_2$.
\item All logarithms are base $2$, unless said otherwise.
\item The symbol $\perp$ is used for two different duals: the Hamming dual of a classical code $C \subseteq \F_2^N$ (with respect to $x \cdot y$), and the symplectic complement of a subspace $V \subseteq \F_2^{2N}$ (with respect to the form $\omega$ of Fact~\ref{fact:commutation}). The two live on spaces of different dimensions ($\F_2^N$ vs.\ $\F_2^{2N}$), so the ambient space always determines which is meant.
\item We use $\mathbf{1}$ to denote the indicator function.
\end{itemize}

\subsection{Classical linear codes}\label{sec:classical}

A (binary) \emph{linear code} of length $N$ and dimension $k$ (written as an $[N, k]$ code) is a $k$-dimensional subspace $C \subseteq \F_2^N$; it elements are \emph{codewords}.
The \emph{Hamming weight} of $x \in \F_2^N$ is $\wtH(x) := |\{ i : x_i \neq 0 \}|$, and the \emph{minimum distance} of $C$ is
\[
\dist(C) \;:=\; \min\{ \wtH(c) : c \in C \setminus \{0\} \};
\]
we also write $\lambda(C) := \dist(C)$. An $[N, k]$ code with $\dist(C) = d$ is an $[N, k, d]$ code. The \emph{dual code} is $C^\perp := \{ x \in \F_2^N : x \cdot c = 0 \ \forall c \in C \}$; recall $\dim C + \dim C^\perp = N$ and $(C^\perp)^\perp = C$.

A code may be \emph{presented} by a \emph{generator matrix} (a matrix whose rows span $C$) or by a \emph{parity-check matrix} $H \in \F_2^{(N - k) \times N}$ (a matrix whose rows span $C^\perp$, so that $C = \{ x : H x = 0 \}$); the two presentations are interconvertible in time $O(N^3)$ by Gaussian elimination. Unless stated otherwise, classical codes in this paper are presented by parity-check matrices.

\subsection{Pauli operators, stabilizer codes, and error detection}\label{sec:pauli}

The single-qubit \emph{Pauli operators} are
\[
X \;=\; \begin{pmatrix} 0 & 1 \\ 1 & 0 \end{pmatrix},
\qquad
Z \;=\; \begin{pmatrix} 1 & 0 \\ 0 & -1 \end{pmatrix},
\]
acting on $\C^2$; they are Hermitian, unitary, square to $I$, and satisfy $ZX = -XZ$.

\begin{definition}[Pauli group]\label{def:pauli}
For $a, b \in \F_2^N$ let $X(a) := \bigotimes_i X^{a_i}$ and $Z(b) := \bigotimes_i Z^{b_i}$. The $N$-qubit Pauli group is $\cP_N := \{ i^{\gamma} X(a) Z(b) : \gamma \in \mathbb{Z}_4, \ a, b \in \F_2^N \}$; the representation $E = i^\gamma X(a) Z(b)$ is unique. The \emph{symplectic part} of $E$ is $\sigma(E) := (a \mid b) \in \F_2^{2N}$. The \emph{projective Pauli group} is $\bP_N := \cP_N / \langle i I \rangle$; the map $\sigma$ induces a group isomorphism $\bP_N \cong (\F_2^{2N}, \oplus)$.
\end{definition}

\begin{fact}\label{fact:commutation}
$Z(b) X(a') = (-1)^{b \cdot a'} X(a') Z(b)$, hence
\[
X(a) Z(b) \cdot X(a') Z(b') \;=\; (-1)^{b \cdot a'} \, X(a \oplus a') Z(b \oplus b'),
\]
and two Paulis $E, F$ commute iff $\omega(\sigma(E), \sigma(F)) = 0$, where
\[
\omega\big((a \mid b), (a' \mid b')\big) \;:=\; a \cdot b' \,\oplus\, b \cdot a'
\]
is the standard symplectic form on $\F_2^{2N}$; otherwise they anticommute.
The form $\omega$ is nondegenerate: if $\omega(e, f) = 0$ for all $f$, then $e = 0$.
\end{fact}

\begin{definition}[weights]\label{def:weights}
For a vector $(a \mid b)$, its symplectic weight is $\wtS(a \mid b) := |\{ i : (a_i, b_i) \neq (0,0) \}|$. The weight of a Pauli is $\mathrm{wt}(E) := \wtS(\sigma(E))$, which is also the number of qubits on which $E$ acts nontrivially. Weight is a function of the projective class; that is, phases do not affect the weight.
\end{definition}

For $V \subseteq \F_2^{2N}$ we write $V^\perp := \{ f : \omega(v, f) = 0 \ \forall v \in V \}$. By nondegeneracy, $\dim V^\perp = 2N - \dim V$ and $(V^\perp)^\perp = V$ for subspaces $V$.

\begin{definition}\label{def:stabilizer}
A \emph{stabilizer group} on $N$ qubits is a subgroup $\cS \le \cP_N$ that is abelian, satisfies $-I \notin \cS$, and is generated by $r$ elements $S_1, \dots, S_r$ whose projective images are linearly independent in $\F_2^{2N}$. Its \emph{code space} is $Q := \{ \ket{\psi} : S \ket{\psi} = \ket{\psi} \ \forall S \in \cS \}$, with projector
\[
\Pi \;:=\; \prod_{j=1}^{r} \frac{I + S_j}{2} \;=\; 2^{-r} \sum_{S \in \cS} S .
\]
We write $\bS := \sigma(\cS) \subseteq \F_2^{2N}$; it is an $r$-dimensional subspace, isotropic ($\bS \subseteq \bS^\perp$) since $\cS$ is abelian. 
We call $Q$ an $[[N, \kappa]]$ code with $\kappa := N - r$ \emph{logical qubits}.
\end{definition}

\begin{fact}\label{fact:dim-sign} Ley $\cS$ be a stabilizer group on $N$ qubits.
\leavevmode
\begin{enumerate}[(a)]
\item $\sigma|_{\cS}$ is injective, $|\cS| = 2^r$, and $\dim Q = \tr \Pi = 2^{N-r} = 2^\kappa$.
\item Let $v_1, \dots, v_r \in \F_2^{2N}$ be linearly independent and pairwise $\omega$-orthogonal, let $P_j$ be any Pauli with $\sigma(P_j) = v_j$ and $P_j^2 = I$, and let $s \in \{\pm 1\}^r$ be any sign vector. Then $\langle s_1 P_1, \dots, s_r P_r \rangle$ is a stabilizer group with projective image $\spn(v_1, \dots, v_r)$, and its code space has dimension $2^{N-r}$.
\end{enumerate}
\end{fact}

\begin{fact}\label{fact:normalizer}
For a stabilizer group $\cS$, let $N(\cS) := \{ P \in \cP_N : P \cS P^\dagger = \cS \}$ and $C(\cS) := \{ P : P S = S P \ \forall S \in \cS \}$. Then $N(\cS) = C(\cS)$ and $\sigma(N(\cS)) = \bS^\perp$.
\end{fact}

\begin{definition}\label{def:detection}
A Pauli $E$ is \emph{detected} by $Q$ if $\Pi E \Pi = c_E \, \Pi$ for some scalar $c_E \in \C$. Otherwise, $E$ is \emph{undetectable}.
\end{definition}

\begin{fact}[Knill--Laflamme for stabilizer codes]\label{fact:trichotomy}
Let $Q$ be a stabilizer code with group $\cS$, and let $E \in \cP_N$ with $e := \sigma(E)$.
\begin{enumerate}[(1)]
\item If $e \notin \bS^\perp$, then $\Pi E \Pi = 0$: $E$ is detected.
\item If $e \in \bS$, then $E = i^\gamma S$ for some $S \in \cS$, $\gamma \in \mathbb{Z}_4$, so $\Pi E \Pi = i^\gamma \Pi$: $E$ is detected and acts on $Q$ as a global phase.
\item If $e \in \bS^\perp \setminus \bS$, then $E$ preserves $Q$, is undetectable, and acts on $Q$ as a \emph{nontrivial} logical operator (not proportional to $\id_Q$).
\end{enumerate}
\end{fact}

\begin{definition}[quantum distance]\label{def:qdist}
For a stabilizer code $Q$ with $\kappa \ge 1$ logical qubits,
\[
\qdist(Q) \;:=\; \min\{\, \wtS(e) \;:\; e \in \bS^\perp \setminus \bS \,\}.
\]
Note that this is well defined for all $Q$ with $\kappa \ge 1$: $\dim \bS^\perp = N + \kappa > N - \kappa = \dim \bS$, so the set is nonempty. By Fact~\ref{fact:trichotomy} it is exactly the minimum weight of an undetectable, logically nontrivial Pauli. A code is \emph{degenerate} if $\min\{ \wtS(e) : e \in \bS \setminus \{0\} \} < \qdist(Q)$; low-weight stabilizer elements never count toward $\qdist$, by definition.
\end{definition}

A stabilizer code is \emph{presented} by a generator matrix $\big( \sigma(S_1); \dots; \sigma(S_r) \big) \in \F_2^{r \times 2N}$, optionally with a sign vector. By Fact~\ref{fact:dim-sign}(b) and Definition~\ref{def:qdist}, both the code dimension and the distance depend only on $\bS$, not on the signs; all reductions in this paper may therefore output all-$(+1)$ signs. Numbers are encoded in binary; the instance size is $\Theta(rN)$ bits.

\subsection{Graph states and CWS codes over linear classical codes}\label{sec:cws}

Throughout, $G$ is a simple undirected graph on vertex set $[N]$ with adjacency matrix $A = A_G \in \F_2^{N \times N}$: symmetric and zero along the diagonal. Graph states were introduced in \cite{HEB04} (see also the survey \cite{HDERVB06}); the CWS construction is due to \cite{CSSZ09}. We use only the case of a \emph{linear} classical code, in which the CWS code is a stabilizer code \cite{CSSZ09}.

\begin{definition}[graph state {\cite{HEB04}}]\label{def:graphstate}
For $v \in [N]$ let $g_v := X_v \prod_{u \sim v} Z_u$, i.e.\ the Hermitian Pauli with $\sigma(g_v) = (e_v \mid A e_v)$ (and $+1$ sign).
\end{definition}

The following lemma is standard \cite{HEB04,HDERVB06}; we include the proof for completeness.

\begin{lemma}\label{lem:graphstate}
For $a \in \F_2^N$ define $g(a) := \prod_v g_v^{a_v}$.
\leavevmode
\begin{enumerate}[(a)]
\item The $g_v$ pairwise commute, so $g(a)$ is order-independent, Hermitian, $g(a)^2 = I$, and $\sigma(g(a)) = (a \mid A a)$; thus $g(a) = \varepsilon_a X(a) Z(A a)$ for some $\varepsilon_a \in \{\pm 1\}$. 
\item $\cS_G := \{ g(a) : a \in \F_2^N \}$ is a stabilizer group with $r = N$ generators $g_1, \dots, g_N$ and $\bS_G = \{ (a \mid A a) : a \in \F_2^N \}$, an $N$-dimensional isotropic subspace.
\item The code space of $\cS_G$ is one-dimensional, spanned by a unit vector $\ket{G}$ (\emph{the graph state}), and $g(a) \ket{G} = \ket{G}$ for all $a$.
\end{enumerate}
\end{lemma}

\begin{proof}
(a) $\omega\big( (e_v \mid A e_v), (e_u \mid A e_u) \big) = e_v \cdot A e_u \oplus A e_v \cdot e_u = A_{vu} \oplus A_{uv} = 0$ by symmetry, so the $g_v$ commute (Fact~\ref{fact:commutation}); products of commuting Hermitian involutions are Hermitian involutions, and symplectic parts add: $\sigma(g(a)) = \bigoplus_v a_v (e_v \mid A e_v) = (a \mid A a)$.

(b) The map $a \mapsto (a \mid A a)$ is linear and injective, so the images of $g_1, \dots, g_N$ are independent and $|\cS_G| = 2^N$; isotropy is the computation in (a) extended by bilinearity: $\omega( (a \mid A a), (a' \mid A a') ) = a \cdot A a' \oplus A a \cdot a' = 0$ using $A^\top = A$. Closure: $g(a) g(a')$ is a product of $g_v$'s with symplectic part $(a \oplus a' \mid A(a \oplus a'))$, hence equals $g(a \oplus a')$. Finally $-I \notin \cS_G$: an element with trivial symplectic part has $a = 0$, and $g(0) = I$ (the empty product).

(c) Fact~\ref{fact:dim-sign}(a) with $r = N$ gives dimension $2^{N - N} = 1$; each $g(a)$ is a product of the generators, hence fixes $\ket{G}$.
\end{proof}

\begin{lemma} [\cite{HDERVB06}]\label{lem:trace}
For every Pauli $P$: $\bra{G} P \ket{G} \neq 0$ only if $\sigma(P) \in \bS_G$. In particular, for $v \in \F_2^N$, $\bra{G} Z(v) \ket{G} = \delta_{v, 0}$.
\end{lemma}

\begin{proof}
$\ket{G}\bra{G} = \Pi_G = 2^{-N} \sum_a g(a)$, so $\bra{G} P \ket{G} = \tr( P \, \Pi_G ) = 2^{-N} \sum_a \tr( P \, g(a) )$. Each summand is the trace of a Pauli with symplectic part $\sigma(P) \oplus (a \mid A a)$, which vanishes unless that part is $0$, i.e.\ unless $\sigma(P) = (a \mid A a) \in \bS_G$. For $P = Z(v)$:
$\sigma(P) = (0 \mid v) \in \bS_G$ forces $a = 0$ and then $v = A \cdot 0 = 0$; and $\braket{G}{G} = 1$.
\end{proof}

\begin{definition}[CWS code {\cite{CSSZ09}}]\label{def:cws}
Let $C \subseteq \F_2^N$ be an $[N, k]$ linear code. Define
\[
\CWS(G, C) \;:=\; \spn\{\, Z(c) \ket{G} \;:\; c \in C \,\}.
\]
Let $h_1, \dots, h_{N-k}$ be a basis of $C^\perp$ (e.g.\ the rows of a full-rank parity-check matrix $H$ of $C$), and define the group $\cS_{G,C} := \langle g(h_1), \dots, g(h_{N-k}) \rangle = \{ g(h) : h \in C^\perp \}$.
\end{definition}

The following proposition is essentially present in \cite{CSSZ09} (in a slightly different language); we include the short proof sketches for completeness.

\begin{proposition}[CWS codes are stabilizer codes; dimension is graph-independent]\label{prop:cws}
For \emph{every} graph $G$ on $[N]$ and every $[N, k]$ code $C$:
\begin{enumerate}[(a)]
\item $\{ Z(c) \ket{G} \}_{c \in C}$ is an orthonormal family; hence $\dim \CWS(G, C) = 2^k$.
\item $\cS_{G,C}$ is a stabilizer group with $N - k$ independent generators and $\bS_{G,C} = \{ (h \mid A h) : h \in C^\perp \}$.
\item $\CWS(G, C)$ equals the code space of $\cS_{G,C}$. Thus $\CWS(G, C)$ is an $[[N, k]]$ stabilizer code: $\kappa = \dim C$, for every graph.
\item Given $A_G$ and $H$, a stabilizer presentation of $\CWS(G, C)$ is computable in time $O(N^3)$.
\end{enumerate}
\end{proposition}

\begin{proof}
(a) $Z(c)^\dagger Z(c') = Z(c \oplus c')$, so $\bra{G} Z(c)^\dagger Z(c') \ket{G} = \delta_{c, c'}$ by Lemma~\ref{lem:trace}. Orthonormal vectors are independent, hence the dimension becomes $2^k$.

(b) Restriction of Lemma~\ref{lem:graphstate}(b) to the subgroup indexed by $C^\perp$; the sub-basis $\{h_j\}$ has independent images.

(c) Each basis vector is fixed: for $h \in C^\perp$ and $c \in C$,
\[
g(h) \, Z(c) \, \ket{G} \;=\; (-1)^{h \cdot c} \, Z(c) \, g(h) \, \ket{G}
\;=\; Z(c) \ket{G},
\]
using Fact~\ref{fact:commutation} ($X(h)$ commutes past $Z(c)$ at cost $(-1)^{h \cdot c} = 1$) and Lemma~\ref{lem:graphstate}(c). So $\CWS(G, C) \subseteq \Fix(\cS_{G,C})$; by (a) the left side has dimension $2^k$, and by Fact~\ref{fact:dim-sign}(a) the right side has dimension $2^{N - (N - k)} = 2^k$; they coincide.

(d) Gaussian elimination to obtain a full-rank $H$, then $N - k$ matrix--vector products. Validity of the all-$(+)$ sign choice is Fact~\ref{fact:dim-sign}(b); choosing different signs may change the code space by a Pauli frame but not $\bS_{G,C}$, hence not the distance (Definition~\ref{def:qdist}).
\end{proof}

\begin{definition}[error classicalization map] {\cite{CSSZ09}}]\label{def:cl}
For a Pauli $E$ with $\sigma(E) = (a \mid b)$, the \emph{induced classical error} is
\[
\Cl_G(E) \;:=\; b \,\oplus\, A_G \, a \;\in\; \F_2^N,
\] 
a function of the projective class of $E$, computable in time $O(N^2)$ from $(a \mid b)$ and $A_G$.
\end{definition}

\begin{lemma}[normalizer test]\label{lem:normtest}
$\sigma(E) \in \bS_{G,C}^{\,\perp}$ if and only if  $\Cl_G(E) \in C$.
\end{lemma}

\begin{proof}
For every $h \in C^\perp$:
\[
\omega\big( (a \mid b), (h \mid A h) \big)
\;=\; a \cdot A h \,\oplus\, b \cdot h
\;=\; (A a) \cdot h \,\oplus\, b \cdot h
\;=\; \big( b \oplus A a \big) \cdot h
\;=\; \Cl_G(E) \cdot h,
\]
using $A^\top = A$ in the second equality. Thus $\sigma(E)$ is $\omega$-orthogonal to all of $\bS_{G,C}$ iff $\Cl_G(E) \cdot h = 0$ for all $h \in C^\perp$, i.e.\ iff $\Cl_G(E) \in (C^\perp)^\perp = C$.
\end{proof}

The following lemma is essentially the error-detection condition of \cite[Thm.~3]{CSSZ09}, (see also \cite[Fact~1]{KK23}), restated in the language of \S\ref{sec:pauli}: as a criterion over $\bS_{G,C}^{\,\perp}$, with the undetected errors classified into stabilizer elements and nontrivial logical operators.
This is the form in which the criterion is used throughout the paper; in this form, it follows in a few lines from Fact~\ref{fact:trichotomy} and Lemma~\ref{lem:normtest}, and we include the proof for completeness.

\begin{lemma}[exact detection criterion for CWS codes]\label{lem:exact}
Let $Q = \CWS(G, C)$ with $k \geq 1$, and let $E$ be a Pauli with $\sigma(E) = (a \mid b)$. Exactly one of the following holds:
\begin{enumerate}[(1)]
\item \textbf{$\Cl_G(E) \notin C$:} $\ \Pi E \Pi = 0$; $E$ is detected.
\item \textbf{$\Cl_G(E) \in C \setminus \{0\}$:} $E$ is undetectable and acts on $Q$ as a nontrivial logical operator. 
\item \textbf{$\Cl_G(E) = 0$} (equivalently $b = A a$):
\begin{enumerate}[(3a)]
\item if $a \in C^\perp$: $E$ is (a phase times) an element of $\cS_{G,C}$; hence has a trivial action on the code;
\item if $a \notin C^\perp$: $E$ is undetectable and logically nontrivial.
\end{enumerate}
\end{enumerate}
Consequently,
\begin{equation}\label{eq:qdist-formula}
\qdist(\CWS(G, C)) \;=\; \min\Big\{ \wtS(a \mid b) \;:\; \underbrace{b \oplus A a \in C \setminus \{0\}}_{\text{class (2)}}
\ \text{ or } \ \underbrace{\big( b = A a \ \text{and} \ a \notin C^\perp \big)}_{\text{class (3b)}} \Big\}.
\end{equation}
\end{lemma}

\begin{proof}
By Lemma~\ref{lem:normtest}, case (1) is $\sigma(E) \notin \bS_{G,C}^{\perp}$, and Fact~\ref{fact:trichotomy}(1) applies. In cases (2) and (3), $\sigma(E) \in \bS_{G,C}^\perp$. Membership in $\bS_{G,C}$ itself means $(a \mid b) = (h \mid A h)$ for some $h \in C^\perp$, i.e.\ exactly: $b = A a$ \emph{and} $a \in C^\perp$. In case (2), $\Cl_G(E) \neq 0$ gives $b \neq A a$, so $\sigma(E) \in \bS_{G,C}^\perp \setminus \bS_{G,C}$ and
Fact~\ref{fact:trichotomy}(3) applies; no property of $G$ beyond $A^\top = A$ and $A_{vv} = 0$ was used anywhere, so this holds for every graph.
In case (3a), $\sigma(E) \in \bS_{G,C}$ and
Fact~\ref{fact:trichotomy}(2) applies. In case (3b), $b = A a$ but $a \notin C^\perp$, so $\sigma(E) \in \bS_{G,C}^\perp \setminus \bS_{G,C}$ and
Fact~\ref{fact:trichotomy}(3) applies. Equation~\eqref{eq:qdist-formula} is
Definition~\ref{def:qdist} with $\bS^\perp \setminus \bS$ decomposed into classes (2) and (3b).
\end{proof}

\begin{remark}\label{rem:fact21}
The correspondence with the criterion as stated in \cite[Fact~1]{KK23} is the following.
Under $\Cl_G(E) \neq 0$, their condition (i) states: detection iff $\Cl_G(E)$ is not a nonzero codeword. 
This is the dichotomy between classes (1) and (2).
Under $\Cl_G(E) = 0$, their condition (ii) states detection iff $E$ commutes with every word operator $Z(c)$, $c \in C$.
This is the split (3a)/(3b), since $X(a)Z(b)$ commutes with $Z(c)$ iff $a \cdot c = 0$ (Fact~\ref{fact:commutation}), i.e.\ iff $a \in C^\perp$.
Following \cite[Def.~2]{KK23}, we define the \emph{graph state distance} $\Gdist(G) := \min\{ \wtS(a \mid A a) : a \neq 0 \}$.
Every error with $a \neq 0$ and $\Cl_G(E) = 0$ has $b = A a$ and hence weight at least
$\Gdist(G)$, a bound over a \emph{superset} of class (3b), since $a \notin C^\perp$ implies $a \neq 0$.
All lower-bound arguments in this paper rule out class (3b) through this superset bound (diagonality at level $w^*$ eliminates everything with $a \neq 0$, $b = A a$ below weight $w^*$), so the finer split is never needed for soundness; its role is to localize degeneracy.
Class (3a) is exactly the set of low-weight stabilizer elements of $\CWS(G, C)$, exactly where degeneracy of the code lives, and Definition~\ref{def:qdist} excludes it by construction, so no separate degenerate-case analysis will be needed anywhere in the paper.
\end{remark}

\begin{remark}\label{rem:onesided}
Two consequences of this subsection hold for every graph $G$ and are invoked repeatedly by the one-sided and derandomized reductions.
\begin{enumerate}[(a)]
\item $\CWS(G, C)$ is a valid $[[N, k]]$ code (Proposition~\ref{prop:cws});
\item taking $c \in C$ of minimum weight and $E = Z(c)$:
$\sigma(E) = (0 \mid c)$, $\Cl_G(E) = c \in C \setminus \{0\}$, so by class (2) of Lemma~\ref{lem:exact}, $E$ is undetectable and nontrivial with $\wtS(E) = \wtH(c)$; hence
\begin{equation}\label{eq:upper-bound}
\qdist(\CWS(G, C)) \;\le\; \dist(C) \qquad \text{for every graph } G.
\end{equation}
This is the upper-bound step in the proof of \cite[Thm.~1]{KK23}; we present it as a standalone inequality because later sections apply it to graphs for which nothing else is known, in particular to the bad graphs in a candidate list.
\end{enumerate}
\end{remark}
\subsection{The problems}
\label{subsec:problems}

\begin{problem}[$\mathrm{MDP}_{n,m}$]
\label{prob:mdp}
The \emph{Minimum Distance Problem} $\mathrm{MDP}_{n,m}$ is the following problem.
\begin{description}
   \item[\textnormal{INSTANCE:}] a parity-check matrix
     $H \in \mathbb{F}_2^{(m-n)\times m}$ of a binary code   $C(H) \subseteq \mathbb{F}_2^{m}$ (of rank $n$), and an integer $t > 0$.
   \item[\textnormal{YES:}] $\mathrm{dist}(C(H)) \le t$.
   \item[\textnormal{NO:}] $\mathrm{dist}(C(H)) \ge t+1$.
\end{description}
\end{problem}

\begin{problem}[$\gamma$-$\mathrm{MDP}_{n,m}$; {\cite[Def.~2.2]{SV19}}]
\label{prob:gapmdp}
For an approximation factor $\gamma \ge 1$, the \emph{Minimum Distance Problem with multiplicative gap $\gamma$}, denoted $\gamma$-$\mathrm{MDP}_{n,m}$, is the following promise problem.
\begin{description}
   \item[\textnormal{INSTANCE:}] a parity-check matrix  $H \in \mathbb{F}_2^{(m-n)\times m}$ of a binary code  $C(H) \subseteq \mathbb{F}_2^{m}$ (of rank $n$), and an integer $t > 0$.
   \item[\textnormal{YES:}] $\mathrm{dist}(C(H)) \le t$.
   \item[\textnormal{NO:}] $\mathrm{dist}(C(H)) > \gamma\, t$.
\end{description}
Instances with $t < \mathrm{dist}(C(H)) \le \gamma t$ satisfy neither promise. Taking $\gamma = 1$ recovers the exact problem $\mathrm{MDP}_{n,m}$ of Problem~\ref{prob:mdp}.
\end{problem}

The reductions we import from \cite{SV19} pass through the Nearest Codeword Problem as an intermediate step, $k\text{-SAT} \to \mathrm{NCP} \to \mathrm{MDP}$. We never state a result about $\mathrm{NCP}$ directly and use it only as a black box inside the cited reductions; we include it here for completeness.

\begin{problem}[$\mathrm{NCP}_{n,m}$; following {\cite[Def.~2.1]{SV19}}]
\label{prob:ncp}
The \emph{Nearest Codeword Problem} $\mathrm{NCP}_{n,m}$ is the following problem.
\begin{description}
   \item[\textnormal{INSTANCE:}] a generator matrix $C \in \mathbb{F}_2^{m\times n}$ of a binary code of rank $n$, a  target vector $y \in \mathbb{F}_2^{m}$, and an integer $t > 0$.
   \item[\textnormal{YES:}] $\mathrm{dist}(y, C) \le t$.
   \item[\textnormal{NO:}] $\mathrm{dist}(y, C) \ge t+1$.
\end{description}
Here $\mathrm{dist}(y, C) := \min_{x \in \mathbb{F}_2^{n}} \mathrm{wt}_H(Cx - y)$ is the distance from the target to the code.\footnote{Our notation deviates from that of \cite{SV19}, who write $t$ for the target vector and $d$ for the threshold; we reserve $d$ for code distances throughout. The dictionary is $y \leftrightarrow t$, $t \leftrightarrow d$, and, for the locally dense gadgets below, $r^\dagger \leftrightarrow d^\dagger$.} Note that $\mathrm{NCP}$ is naturally presented by a \emph{generator} matrix, whereas $\mathrm{MDP}$ is presented by a parity-check matrix; the two presentations are interconvertible in time $O(m^3)$, so the distinction is immaterial to the reductions.
\end{problem}

\begin{problem}[$\QMDP_{\kappa,N}$]
\label{prob:qmindist}
The \emph{Quantum Minimum Distance Problem} $\mathrm{\QMDP}_{\kappa,N}$ is the following problem.
\begin{description}
   \item[\textnormal{INSTANCE:}] a stabilizer presentation of an $[[N, \kappa]]$ stabilizer code $Q$ with $\kappa \ge 1$ logical  qubits, and an integer $t > 0$.
   \item[\textnormal{YES:}] $\mathrm{qdist}(Q) \le t$.
   \item[\textnormal{NO:}] $\mathrm{qdist}(Q) \ge t+1$.
\end{description}
Here $\kappa$ is the number of logical qubits (the fine-grained parameter) and $N$ is the block length (number of physical qubits); $\mathrm{qdist}(Q)$ is the minimum weight of an undetectable, logically nontrivial Pauli.
\end{problem}

\begin{problem}[$\gamma$-$\mathrm{GapQDist}$]
\label{prob:gapmultqdist}
For $\gamma \ge 1$, the multiplicative variant has the same instances, with YES: $\mathrm{qdist}(Q) \le t$ and NO: $\mathrm{qdist}(Q) > \gamma t$.
\end{problem}

The additive-gap variants below are the form in which our gap results are stated; we follow \cite{KK23, GJS25} in measuring the gap as an additive function of the block length rather than as a multiplicative factor.

\begin{problem}[$\mathrm{GapAddDist}_{\tau}$; cf.\ {\cite{DMS03,
GJS25}}]
\label{prob:gapadddist}
For a constant $\tau > 0$, the \emph{Minimum Distance Problem with additive gap $\tau$} is the following promise problem.
\begin{description}
   \item[\textnormal{INSTANCE:}] a parity-check matrix   $H \in \mathbb{F}_2^{(m-n)\times m}$ of a binary code   $C(H) \subseteq \mathbb{F}_2^{m}$, and an integer $t > 0$.
   \item[\textnormal{YES:}] $\mathrm{dist}(C(H)) \le t$.
   \item[\textnormal{NO:}] $\mathrm{dist}(C(H)) > t + \tau m$.
\end{description}
Instances with $t < \mathrm{dist}(C(H)) \le t + \tau m$ satisfy neither promise.
\end{problem}

\begin{problem}[$\mathrm{GapAddQDist}_{g}$; cf.\ {\cite[Def.~4]{KK23}}]
\label{prob:gapaddqdist}
For a gap function $g\colon \mathbb{N} \to \mathbb{N}$, the \emph{Quantum Minimum Distance Problem with additive gap $g$} is the following promise problem.
\begin{description}
   \item[\textnormal{INSTANCE:}] a stabilizer presentation of an  $[[N, \kappa]]$ stabilizer code $Q$ with $\kappa \ge 1$, and an  integer $t > 0$.
   \item[\textnormal{YES:}] $\mathrm{qdist}(Q) \le t$.
   \item[\textnormal{NO:}] $\mathrm{qdist}(Q) > t + g(N)$.
\end{description}
We write $\mathrm{GapAddQDist}_{\alpha N}$ for the case $g(N) = \alpha N$ with $\alpha > 0$ constant, a gap \emph{proportional to the block length}, the regime of Theorem~\ref{thm:nphard-rand} and $\mathrm{GapAddQDist}_{\tau N^{\epsilon}}$ for the sublinear regimes $\epsilon < 1$ of \cite{KK23}. Taking $g = 0$ recovers $\QMDP_{\kappa, N}$ of Problem~\ref{prob:qmindist}.
\end{problem}

Multiplicative and additive gaps are related but not interchangeable.
\paragraph{additive versus multiplicative gaps:}
\label{rem:gap-forms}
A multiplicative gap $\gamma$ at threshold $t$ is the additive gap $(\gamma - 1)t$, so the two forms coincide up to the \emph{relative} size of the threshold; on instances with $t = \Theta(N)$ a constant multiplicative gap is an additive $\Omega(N)$ gap, while on instances with $t = o(N)$ it is an additive $o(N)$ gap. 
This is why hardness at a constant multiplicative factor does not by itself yield hardness at a linear additive gap, and why the classical bases of \S\ref{sec:bases} are stated additively. Lemma~\ref{lem:addbase} converts an additive classical gap into an additive quantum gap with no condition on the threshold profile, whereas the multiplicative route requires $t = \Theta(m)$. For the same reason, we have approximation hardness for $\qdist$ in additive form: since $\qdist(Q) \le N$ always, an additive gap $\alpha N$ with constant $\alpha$ is the strongest scale at which the promise is not vacuous.

\paragraph*{CSS variants}
Recall that a stabilizer code is \emph{CSS} if it admits a generating set in which every generator is of pure $X$-type ($X(u)$ for some $u \in \F_2^N$) or pure $Z$-type ($Z(v)$ for some $v \in \F_2^N$); such a code is presented by a pair of parity-check matrices $(H_X, H_Z)$ with $H_X H_Z^\top = 0$, the condition being equivalent commutativity of the generators.
\begin{problem}[$\mathrm{CSSMDP}_{\kappa,N}$]
\label{prob:cssqmindist}
\begin{description}
   \item[\textnormal{INSTANCE:}] a CSS presentation $(H_X, H_Z)$,  $H_X H_Z^\top = 0$, of an $[[N, \kappa]]$ CSS code $Q$ with   $\kappa \ge 1$, and an integer $t > 0$.
   \item[\textnormal{YES:}] $\qdist(Q) \le t$.
   \item[\textnormal{NO:}] $\qdist(Q) \ge t+1$.
\end{description}
\end{problem}

\begin{problem}[$\gamma$-$\mathrm{GapCSSMDP}$]
\label{prob:gapmultcssqdist}
For $\gamma \ge 1$:
\begin{description}
   \item[\textnormal{INSTANCE:}] a CSS presentation $(H_X, H_Z)$, $H_X H_Z^\top = 0$, of an $[[N, \kappa]]$ CSS code $Q$ with $\kappa \ge 1$, and an integer $t > 0$.
   \item[\textnormal{YES:}] $\qdist(Q) \le t$.
   \item[\textnormal{NO:}] $\qdist(Q) > \gamma t$.
\end{description}
\end{problem}

\begin{problem}[$\mathrm{GapAddCSSMDP}_{g}$]
\label{prob:gapaddcssqdist}
For a gap function $g \colon \mathbb{N} \to \mathbb{N}$:
\begin{description}
   \item[\textnormal{INSTANCE:}] a CSS presentation $(H_X, H_Z)$,  $H_X H_Z^\top = 0$, of an $[[N, \kappa]]$ CSS code $Q$ with  $\kappa \ge 1$, and an integer $t > 0$.
   \item[\textnormal{YES:}] $\qdist(Q) \le t$.
   \item[\textnormal{NO:}] $\qdist(Q) > t + g(N)$.
\end{description}
\end{problem}

\subsection{SAT problems and hypotheses}
\label{sec:hyps}
 
We recall the SAT problems and fine-grained hypotheses on which the hardness of \cite{SV19} rests. All of them are standard; we state only the binary ($q=2$) versions, which is all we use. 
The reduction from these hypotheses to the coding problems passes through the $q$-ary variants of \cite[\S3]{SV19}; we consume that reduction as a black box (Corollary~\ref{imp:ncp}) and do not reproduce its $q$-ary gadget here. Other than that, certain parameters of this gadget would be important for us later, and we will get back to them in the section that follow.
 
\begin{problem}[$k$-\textup{SAT}]
\label{prob:ksat}
For a constant integer $k \ge 3$, $k$-\textup{SAT} is the following problem.
\begin{description}
   \item[\textnormal{INSTANCE:}] a $k$-CNF formula $\Phi$ on $n$ variables,  i.e.\ a conjunction of clauses each a disjunction of $k$ literals.
   \item[\textnormal{YES:}] $\Phi$ is satisfiable.
   \item[\textnormal{NO:}] $\Phi$ is unsatisfiable.
\end{description}
\end{problem}
 
\begin{problem}[\textup{Max-}$k$-\textup{SAT}]
\label{prob:maxksat}
For a constant integer $k \ge 2$, \textup{Max-}$k$-\textup{SAT} is the following problem.
\begin{description}
   \item[\textnormal{INSTANCE:}] a $k$-CNF formula $\Phi$ with $m$ clauses on $n$ variables, and an integer $r \le m$.
   \item[\textnormal{YES:}] some assignment satisfies at least $r$ clauses.
   \item[\textnormal{NO:}] no assignment satisfies $r$ clauses.
\end{description}
Taking $r = m$ recovers $k$-\textup{SAT}: satisfying all $m$ clauses is ordinary satisfiability. 
\end{problem}
 
\begin{hypothesis}[SETH; {\cite[Def.~2.8]{SV19}}]
\label{imp:sat}
For every constant $\varepsilon > 0$ there is a constant integer $k \ge 3$ such that no $2^{(1-\varepsilon)n}$-time deterministic algorithm solves $k$-\textup{SAT} on $n$ variables. \emph{Randomized} \textup{SETH} (resp.\ \emph{non-uniform} \textup{SETH}) is the same statement for randomized algorithms (resp.\ circuit families).
\end{hypothesis}
 
\begin{hypothesis}[Gap-ETH; {\cite[Def.~2.9]{SV19}}]
\label{imp:gapeth}
There is a constant $s \in (0,1)$ such that no $2^{o(n)}$-time algorithm distinguishes satisfiable $3$-\textup{SAT} instances on $n$ variables (with $O(n)$ clauses) from instances in which every assignment satisfies fewer than an $s$-fraction of the clauses. \emph{Randomized} and \emph{non-uniform} \textup{Gap-ETH} are the analogous statements for randomized algorithms and circuit families.
\end{hypothesis}

\subsection{Results from \cite{SV19}}

\begin{theorem}[$k$-SAT $\to$ NCP;
{\cite[Thm.~4.2]{SV19}}]
\label{imp:ncp_original}
There is a deterministic $\mathrm{poly}(n,m,q^{k})$-time Karp reduction mapping any \textup{Max-}$(q,k)$-\textup{SAT} instance with $m$ clauses on $n\le m$ variables and value $r\le m$ to an \textup{NCP} instance of
\[
   \text{rank } \le n,\qquad
   \text{block length } = (q^{k}-1)\,m,\qquad
   \text{threshold } = (q^{k}-1)m - q^{k-1}r .
\]
Moreover, the reduction is gap-preserving: for any $0 < s \le c \le 1$, it maps $(s,c)$-\textup{Gap-}$(q,k)$-\textup{SAT} to $\gamma$-\textup{NCP} with threshold $t = (q^{k}-1)m - q^{k-1}cm$, where
\[
   \gamma \;=\; \frac{1 - s/q - q^{-k}}{\,1 - c/q - q^{-k}\,} .
\]
\end{theorem}

In our case, we use the following special case of the above theorem with $q=2$, $r=m$, and a constant integer $k$.
This turns \textup{Max-}$(2,k)$-\textup{SAT} into ordinary $k$-\textup{SAT}.

\begin{corollary}
    [$k$-SAT $\to$ NCP;
{\cite{SV19}}]
\label{imp:ncp}
There is a deterministic $\mathrm{poly}(n,m)$-time Karp reduction mapping any $k$-\textup{SAT} instance with $m$ clauses on $n\le m$ variables to an \textup{NCP} instance of
   \[
   \text{rank }\le n,\qquad
   \text{block length } =  O(m).
\]

\end{corollary}

\begin{observation}[the NCP instance has full rank]
\label{lem:fullrank}
Suppose every variable of the input \textup{Max-}$(q,k)$-\textup{SAT} instance occurs in at least one clause. Then the generator matrix $\Phi \in \F_q^{(q^{k}-1)m \times n}$ produced by the reduction of Theorem~\ref{imp:ncp_original} has full column rank; in particular, the output \textup{NCP} instance has rank exactly $n$.
\end{observation}
 
\begin{proof}
The reduction of \cite[Thm.~4.2]{SV19} builds $\Phi$ from clause blocks $\Phi_i \in \F_q^{(q^{k}-1)\times n}$: the columns of $\Phi_i$ indexed by the $k$ variables of the clause $\phi_i$ form the Hadamard gadget $C \in \F_q^{(q^{k}-1)\times k}$ of \cite[Claim~4.1]{SV19}, and all other columns of $\Phi_i$ are zero. The rows of $C$ are all non-zero vectors of $\F_q^{k}$; in particular they include the standard basis vectors, so $C$ has full column rank, i.e.\ $Cv \neq 0$ for every $v \in \F_q^{k}\setminus\{0\}$.
 
Now let $a \in \F_q^{n}\setminus\{0\}$, pick a coordinate $j$ with $a_j \neq 0$, and pick a clause $\phi_i$ in which $x_j$ occurs. Then $\Phi_i a = C\,a|_{\phi_i}$, where $a|_{\phi_i} \in \F_q^{k}$ is the restriction of $a$ to the variables of $\phi_i$. This restriction is non-zero (its coordinate for $x_j$ equals $a_j$), so $\Phi_i a \neq 0$, and hence $\Phi a \neq 0$.
\end{proof}

\begin{theorem}[randomized NCP $\to$ exact MDP;
{\cite[Cor.~5.7]{SV19}}]
\label{imp:rs}
For every constant $\varepsilon' \in \left(0,\tfrac12\right)$ there is an efficient \emph{randomized} $\mathrm{poly}_{\varepsilon'}(m)$-time reduction that maps any \textup{NCP} instance of rank $n (\geq 2)$ and block length $m$ to an exact-\textup{MDP} instance of
\[
   \text{rank } \le (1+\varepsilon')\,n,\qquad
   \text{block length } = \mathrm{poly}_{\varepsilon'}(m).
\]
\end{theorem}
 
\begin{theorem}[deterministic NCP $\to$ exact MDP;
{\cite[Thm.~5.16]{SV19}}]
\label{imp:turing}
For every constant $\varepsilon' \in \left(0,\tfrac12\right)$ there is a \emph{deterministic} (Turing) reduction, running in time
\[
  T \leq  2^{\,3(1+\varepsilon')n/4}\cdot \mathrm{poly}_{\varepsilon'}(n,m),
\]
that maps any \textup{NCP} instance of rank $n$ and block length $m$ to
\[
   Q \;=\; 2^{\,3n/4}\quad\text{exact-\textup{MDP} instances, each of rank}
\]
\[
   n' =  \left\lceil \tfrac{(1+\varepsilon')\,n}{4} \right\rceil + 1 \;\le\; \tfrac{(1+\varepsilon')\,n}{4} + 2, \qquad \text{ and each of block length $= \mathrm{poly}_{\varepsilon'}(m)$}.
\]

\end{theorem}

The following is not stated explicitly in \cite{SV19}, who bound the ambient dimension only by an unspecified $\mathrm{poly}(m)$; it follows by composing \cite[Cor.~5.5 and 5.6]{SV19} along with careful tracking of parameters of the \emph{Reed--Solomon gadget}. 
For our purposes, it does not matter what the exact gadget is; it suffices to treat it as a black box.
\begin{observation}[block length; {\cite{SV19}}]
\label{lem:blocklen}
On the hard instances produced by the Reed--Solomon gadget family of \cite{SV19}, the block length $m'$ satisfies
\[
   m' \;=\; n^{\Theta(1/\varepsilon')}
   \;=\; \mathrm{poly}(n)\qquad\text{for every fixed }\varepsilon' .
\]
In particular, for constant $\varepsilon'$ every factor $\mathrm{poly}(m')$ is $\mathrm{poly}(n) \leq 2^{o(n)}$.
\end{observation}

\begin{proof}
We track the two parameters of the gadget through \cite[Cor.~5.5, 5.6]{SV19}. By \cite[Cor.~5.5]{SV19}, the concatenated Reed--Solomon gadget over $\mathbb{F}_q$ has block length
\begin{equation}\label{eq:cor55}
   m \;=\; \kappa\,(q^{\kappa}-1)^{2} \;=\; \Theta\!\big(\kappa\, q^{2\kappa}\big).
\end{equation}
By \cite[Cor.~5.6]{SV19}, to obtain $M \ge q^{(1-\varepsilon')n}$ locally dense codewords one takes
\begin{equation}\label{eq:cor56}
   \kappa \;=\; 10\Big\lceil \tfrac{1+\log_q n}{\varepsilon'}\Big\rceil
   \;=\; \Theta\!\Big(\tfrac{\log_q n}{\varepsilon'}\Big).
\end{equation}
Substituting \eqref{eq:cor56} into the dominant term of \eqref{eq:cor55} and using $q^{\log_q n} = n$,
\[
   q^{2\kappa}
   \;=\; q^{\,2\cdot\Theta\!\left(\frac{\log_q n}{\varepsilon'}\right)}
   \;=\; \Big(q^{\log_q n}\Big)^{\Theta(1/\varepsilon')}
   \;=\; n^{\Theta(1/\varepsilon')}.
\]
The prefactor $\kappa = \Theta(\log_q n/\varepsilon') = n^{o(1)}$ is subpolynomial and is absorbed into the exponent, so $m = n^{\Theta(1/\varepsilon')}$. 
The reductions of \cite[Cor.~5.7, Thm.~5.16]{SV19} enlarge the block length only polynomially (in $m$ and $q=O(1)$), and a polynomial of $n^{\Theta(1/\varepsilon')}$ is again $n^{\Theta(1/\varepsilon')}$; hence the final MDP block length is $m' = n^{\Theta(1/\varepsilon')}$.
For fixed $\varepsilon'$ this is $\mathrm{poly}(n)$,  and hence, $\mathrm{poly}(m') = \mathrm{poly}(n) \leq  2^{o(n)}$.
\end{proof}

\begin{theorem}[hardness of gap MDP under non-uniform Gap-ETH;
{\cite[Thm.~1.2; \S5.2]{SV19}}]
\label{imp:gapmdp}
There is a constant $\gamma_2 > 1$ such that no $2^{o(n)}$-time algorithm solves $\gamma_2$-$\mathrm{MDP}$ on binary codes of rank $\Theta(n)$, where $n$ is the number of variables of the underlying \textup{Gap-3-SAT} instance, unless non-uniform \textup{Gap-ETH} is false.
\end{theorem}

The reduction of \cite{SV19} establishing Theorem~\ref{imp:gapmdp} is itself non-uniform: it relies on the existence of the kissing-number codes of Ashikhmin--Barg--Vl\u{a}du\c{t} \cite{ABV01}, for which no efficient construction is known \cite[\S5.2]{SV19}. Hence,  the \emph{non-uniform} form of \textup{Gap-ETH} is required.

\section{Upgrading \textup{\cite{SV19}}}
\label{sec:worked}

 \subsection{additional $\mathrm{poly}(m)$ factor}
Our reductions use the classical hardness results of \cite{SV19} as a black box, but not in the exact form in which they are stated.
Here we state the two strengthenings we need.
Both are established by carefully tracking parameters through the reductions of \cite{SV19}. 
The proofs can be found in the appendix, where we state the changes and the parameter tracking required for the reader's convenience.

The first concerns a polynomial factor. \cite{SV19} state their fine-grained lower bound for exact \MDP in the form ``no $2^{(1-\eps)n}$-time algorithm'', with no polynomial factor in the block length attached.
Our reduction preserves the rank exactly but enlarges the block length polynomially, so the composition produces an algorithm carrying such a factor, and we need a form of the classical hardness that rules it out.

\begin{theorem}
\label{thm:strong}
Let $\varepsilon>0$ be a constant. 
Then, there is no algorithm solving $\mathrm{MDP}_{n',m'}$ in time $2^{(1-\varepsilon)\,n'}\cdot\mathrm{poly}(m')$, in either of the following senses:
\begin{enumerate}
   \item[(a)\label{item:strong_a}]  no  such \emph{deterministic} algorithm, unless \textup{SETH} is false;
   \item[(b)] \label{item:strong_b} no such \emph{randomized} algorithm, unless randomized \textup{SETH} is false.
\end{enumerate}
\end{theorem}

\subsection{On parameters for the gap instances}
\label{sec:linprofile}

The second concerns the gap instances.
\cite{SV19} bound the ambient dimension of their hard $\gamma_2$-\MDP instances only by an unspecified $\poly_{\gamma,\eps}(m)$, which is too loose for our use: on these instances, every parameter is in fact linear in the number of Gap-3-SAT variables.
\begin{observation}[parameters of the gap instances]
\label{prop:linprofile}
The hard $\gamma_2$-$\mathrm{MDP}$ instances of Theorem~\ref{imp:gapmdp}
can be taken with
\[
   \text{rank } = \Theta(n), \qquad
   \text{block length } m = \Theta(n), \qquad
    t = \Theta(m),
\]
where $n$ is the number of variables of the underlying
\textup{Gap-3-SAT} instance.
\end{observation}

Our reduction requires all three linearities.
The rank determines the number of logical qubits and the block length determines the number of physical qubits, so when both are $\Theta(n)$ the two are proportional to one another.
The threshold determines the size of the gap: when $t=\Theta(m)$, a constant multiplicative classical gap becomes an additive quantum gap linear in the block length.
Neither consequence would survive if the classical instances merely had $m=\poly(n)$ or $t=o(m)$.
One step of the proof goes beyond what \cite{SV19} state: their gadget radius is bounded only from above, and the matching lower bound is what forces the number of gadget copies to be constant.
We include the proof in the appendix.

\section{Compatible graphs}
\label{sec:blockgraphs}

This section introduces two graph-theoretic ingredients required for the reductions of Section~\ref{sec:finegrained}: the \emph{threshold form} (Lemma~\ref{lem:threshold}), which sandwiches $\qdist$ between the classical distance and the level, and \emph{random compatibility with NP-certified failure} (Lemma~\ref{lem:randomcompat}), which supplies the graphs the construction requires. 
Definition~\ref{def:compatible} introduces compatibility, the condition on which every result in this paper rests; Section~\ref{sec:compatdef} explains how it differs from the global condition used in earlier work. Everything here uses only the detection criterion of Lemma~\ref{lem:exact} and elementary first-moment arguments.

Throughout this section, $C'$ is an $[N, k']$ binary linear code with $\lambda := \dist(C')$, presented by a full-rank parity-check matrix $H' \in \F_2^{(N-k') \times N}$, and $\varepsilon_1 > 0$ is a constant small enough for the bounds of Lemma~\ref{lem:randomcompat} to hold; the admissible range is fixed in Section~\ref{sec:finegrained}.

\begin{claim}\label{claim:half_uniformity}
    Let $\Phi_x$ be a map from graphs to $\mathbb{F}_2^n$ as follows: $
\Phi_x(\G) = A_\G x .
$
Then for any $y$, 
\begin{equation*} \label{eq:phi_distribution}
    \Pr_{\G }[\Phi_x(\G) = y]
=
\begin{cases}
0, & \text{if } \langle y, x \rangle \neq 0, \\[6pt]
2^{-(n-1)}, & \text{if } \langle y, x \rangle = 0 .
\end{cases}
\end{equation*}
That is,  for any non-zero $x$, $\Phi_x(\G)$ is uniformly distributed over the orthogonal subspace of $x$.
\end{claim}
\begin{proof}
Let $V = \{A \in \mathbb{F}_2^{n\times n} : A \text{ symmetric},\ A_{ii}=0\}$, now $\Phi_x : V \to \mathbb{F}_2^n$ is $\mathbb{F}_2$-linear. Since $\Phi_x$ is linear, its nonempty fibres are cosets of $\ker \Phi_x$ and hence all of the same size. So $\Phi_x(\G)$ is uniform on $\operatorname{im}\Phi_x$ for $\G$ uniform, and it suffices to show $\operatorname{im}\Phi_x = x^\perp$.

\medskip \noindent 
For the inclusion $\operatorname{im}\Phi_x \subseteq x^\perp$, note that for any $A \in V$,
\[
\langle Ax, x\rangle = x^\top\! A x = \sum_{i\neq j} x_i A_{ij} x_j + \sum_i A_{ii}x_i = 0,
\]
since the off-diagonal terms cancel in pairs by symmetry and $A_{ii}=0$.

\medskip \noindent 
For the reverse inclusion, fix $y \in x^\perp$ and choose $p$ with $x_p = 1$. Define $A \in V$ by setting $A_{ip} = A_{pi} = 1$ for every $i \neq p$ with $y_i = 1$, and all other entries zero. For $i \neq p$ we get $(Ax)_i = A_{ip}x_p = y_i$, and
\[
(Ax)_p = \sum_{i \neq p} A_{pi}x_i = \sum_{\substack{i \neq p \\ y_i = 1}} x_i = \langle y,x\rangle + y_p x_p = y_p ,
\]
using $\langle y, x\rangle = 0$ and $x_p = 1$. Hence $Ax = y$, so $\operatorname{im}\Phi_x = x^\perp$, which has size $2^{n-1}$.
\end{proof}

\subsection{Compatibility and diagonality}
\label{sec:compatdef}

\begin{definition}[compatible; diagonal]
\label{def:compatible}
Let $C \subseteq \F_2^{n}$ be a linear code, $G$ a graph on $n$ vertices with adjacency matrix $A$, and $w \ge 1$ an integer (the \emph{level}).
\begin{enumerate}
    \item The graph $G$ is \emph{$(C, w)$-compatible} if for every Pauli $E$ with $\sigma(E) = (a \mid b)$, $a \neq 0$, and $\wtS(a \mid b) < w$,
\[
   \Cl_G(E) \;\notin\; C \setminus \{0\}.
\]
\item 
It is \emph{$w$-diagonal} if $\Gdist(G) \ge w$, i.e.\ (with the graph state distance of Remark~\ref{rem:fact21}) every $a \neq 0$ satisfies $\wtS(a \mid A a) \ge w$.
\end{enumerate}

\end{definition}
Definition~\ref{def:compatible} is extracted from the detection criterion of Lemma~\ref{lem:exact}, but is not a restatement of it: the criterion classifies errors given $(G, C)$, whereas compatibility specifies graphs given $(C, w)$. Three deliberate choices separate the two.

First, only the nonzero-syndrome clause is kept; the zero-syndrome errors are delegated to the code-independent property of diagonality. The probabilistic analysis of Section~\ref{sec:randomcompat} and the derandomization both follow this split exactly.

Second, pure-$Z$ errors ($a = 0$) are excluded. There, the graph drops out of $\Cl_G(E) = b$, and the pure-$Z$ errors at codewords are precisely the intended logicals carrying $\dist(C)$ into $\qdist$ (Remark~\ref{rem:onesided}(b)); forbidding them would be unsatisfiable in the YES case. Their weight is instead controlled by $\dist(C)$ itself, in the proof of Lemma~\ref{lem:threshold}.

Third, the condition is truncated at a level $w$. This is what makes it achievable by a random graph, monotone in $w$, and refutable by an $O(n)$-bit certificate, the three properties required for our proofs.

The same criterion admits other extractions. \cite{KK23} keep the zero-syndrome clause as the global requirement $\Gdist(\G) \ge \tau$, and handle the nonzero-syndrome errors instead by a \emph{weight} argument, showing that $\Cl_G(E)$ is too heavy to be a codeword. 
Their argument requires two things compatibility does not: the graph must be sparse, with any two columns of $A$ sharing at most one coordinate, that is, a $4$-cycle-free graph, and every codeword must be light-weight, which forces the classical code into a vanishing fraction of the block length. Compatibility replaces the weight argument by a probabilistic one, and imposes neither condition.

\begin{remark}[level monotonicity]
\label{rem:monotone}
Both conditions are universally quantified over errors of weight below the level, so they are monotone in $w$: if $G$ is $(C, w)$-compatible (resp.\ $w$-diagonal) and $w' \le w$, then $G$ is $(C, w')$-compatible (resp.\ $w'$-diagonal). Reductions may therefore construct or sample a graph at one level $w^*$ and invoke the properties at any lower level.
\end{remark}

\subsection{The threshold form}
\label{sec:threshold}

The next lemma is the form in which every reduction of Section~\ref{sec:finegrained} uses the CWS machinery. Part~(a) holds for \emph{every} graph, good or bad; this one-sidedness is what the truth-table reductions of Section~\ref{sec:hypD} exploit. Part~(b) is the two-sided bracket, provable conditioned on a good graph. Following Remark~\ref{rem:fact21}, the degenerate class~(3a) of Lemma~\ref{lem:exact} needs no separate treatment, and class~(3b) is excluded through the superset bound by diagonality.

\begin{lemma}[threshold form]
\label{lem:threshold}
Let $C'$ be an $[N, k', \lambda]$ code and let $t, g \ge 0$ be integers, and set $w^* := t + g + 1$.
For $\G \in \mathcal{G}_N$ with adjacency matrix $A_\G$, define $Q_\G := \CWS(\G, C')$.
\begin{enumerate}[(a)]
\item For \emph{every} $\G$: $Q_\G$ is an $[[N, k']]$ stabilizer code, a presentation of which is computable in time $O(N^3)$ from $(A_\G, H')$, and
\[
   \qdist(Q_\G) \;\le\; \lambda,
\]
witnessed by $E^* = Z(c)$ at a minimum-weight $c \in C'$. Explicitly, $E^*$ is of \emph{pure $Z$-type}, is undetectable and logically nontrivial for every $\G$, and has $\wtS(E^*) = \lambda$.
\item If $\G$ is $(C', w^*)$-compatible and $w^*$-diagonal, then
\begin{enumerate}[(b1)]
\item for $\lambda \le t$: $\qdist(Q_\G) = \lambda$;
\item for $\lambda > t + g$: $\qdist(Q_\G) \ge w^* = t + g + 1$.
\end{enumerate}
\end{enumerate}
\end{lemma}

\begin{proof}
(a) Dimensions and computability are Proposition~\ref{prop:cws} applied to the graph $\G$ on $N$ vertices and the $[N, k']$ code $C'$; the dimension is graph-independent. The upper bound and the witness are Remark~\ref{rem:onesided}(b): for a minimum-weight $c \in C'$, the pure-$Z$ Pauli $E^* = Z(c)$ has $\sigma(E^*) = (0 \mid c)$ and $\Cl_{\G}(E^*) = c \in C' \setminus \{0\}$, so by class~(2) of Lemma~\ref{lem:exact} it is undetectable and logically nontrivial for every graph, with $\wtS(E^*) = \wtH(c) = \lambda$.

\smallskip\noindent 
(b) By eq.~\eqref{eq:qdist-formula}, $\qdist(Q_\G)$ is the minimum weight over classes (2) and (3b) of Lemma~\ref{lem:exact}. We bound each class from below on a \emph{good} graph. Let $e = (a \mid b)$.

\smallskip
\noindent\emph{Class (3b): $a \neq 0$ and $b = A_\G a$.} Such $e$ has $\wtS(e) = \wtS(a \mid A_\G a) \ge \Gdist(\G) \ge w^*$ by $w^*$-diagonality; this is the superset bound of Remark~\ref{rem:fact21}.

\smallskip
\noindent\emph{Class (2) with $a \neq 0$.} If $\wtS(e) < w^*$, then $(C', w^*)$-compatibility gives $\Cl_{\G}(E) \notin C' \setminus \{0\}$, so $e$ is not in class (2), a contradiction. Hence $\wtS(e) \ge w^*$.

\smallskip
\noindent\emph{Class (2) with $a = 0$.} Then $\Cl_{\G}(E) = b \in C' \setminus \{0\}$, so $\wtS(e) = \wtH(b) \ge \lambda$.

\smallskip
\noindent Combining the three cases, $\qdist(Q_\G) \ge \min(w^*, \lambda)$.

\smallskip
\noindent If $\lambda > t + g$, then $\lambda \ge t + g + 1 = w^*$ since $\lambda$ is an integer, so $\qdist(Q_\G) \ge \min(w^*, \lambda) = w^*$.

\smallskip
\noindent If $\lambda \le t$, then $\lambda < w^*$. 
So $\qdist(Q_\G) \ge \min(w^*, \lambda) = \lambda$; combining with part~(a) gives $\qdist(Q_\G) = \lambda$.
\end{proof}

The explicit clause in~(a) that the YES-side witness is of pure $Z$-type is consumed by the conversion to CSS codes. 
The doubling map of \cite{BTL10} is lossless on pure-$Z$ logicals, which is what keeps the additive gap after conversion. The clause is also graph-independent, which is what makes the one-sided selection in Section~\ref{sec:hypD} valid on bad candidates.

\subsection{Random graphs are good}
\label{sec:randomcompat}

\begin{lemma}[random compatibility, with NP-certified failure]
\label{lem:randomcompat}
Let $C'$ be an $[N, k']$ code with $k' \le \frac{N}{4}$, and let $w^* \le \varepsilon_1 N$ for a small enough constant $\varepsilon_1 > 0$. For $\G$ uniform in $\mathcal{G}_N$,
\[
   \Pr\big[\, \G \text{ is not } (C', w^*)\text{-compatible or not } w^*\text{-diagonal} \,\big] \;\le\; 2^{-\delta(\varepsilon_1) N},
\]
for some constant $\delta(\varepsilon_1) > 0$. Moreover, failure is witnessed by a vector $(a \mid b) \in \F_2^{2N}$ of $O(N)$ bits, verifiable in time $O(N^2)$ given $A_\G$ and $H'$.
\end{lemma}
\begin{proof}
\emph{Compatibility.} Fix $e = (a \mid b)$ with $a \neq 0$ and $\wtS(e) < w^*$. By Claim~\ref{claim:half_uniformity}, $A_\G a$ is uniform on $a^\perp$, so $\Cl_\G(E) = b \oplus A_\G a$ is uniform on the coset $b \oplus a^\perp$, which has size $2^{N-1}$. Hence
\[
   \Pr\big[\Cl_\G(E) \in C'\big] \;=\; \frac{\big|C' \cap (b \oplus a^\perp)\big|}{2^{N-1}} \;\le\; 2^{k'-N+1} \;\le\; 2^{-\frac{3N}{4}+1},
\]
using $k' \le \frac{N}{4}$. The number of $e \in \F_2^{2N}$ with $\wtS(e) < w^* \le \varepsilon_1 N$ is at most $\sum_{w<w^*}\binom{N}{w}3^w$, so by a union bound
\begin{align*}
    \Pr\big[\G \text{ is not } (C', w^*)\text{-compatible}\big]
    & \le 2^{-\frac{3N}{4}+1} \sum_{w<w^*}\binom{N}{w}3^{w} \\
    & \le 2^{-\frac{3N}{4}+1}\, w^*\left(\frac{3eN}{w^*}\right)^{w^*} \\
    & \le 2^{-\frac{3N}{4}+1} \cdot \varepsilon_1 N \left(\frac{3e}{\varepsilon_1}\right)^{\varepsilon_1 N} \\
    & \le 2^{-\delta N} \qquad \text{for some constant } \delta > 0 .
\end{align*}

\smallskip\noindent
\emph{Diagonality.} Suppose $a \neq 0$ satisfies $\wtS(a \mid A_\G a) < w^*$; then in particular $\wtH(a) < w^*$ and $\wtH(A_\G a) < w^*$. By Claim~\ref{claim:half_uniformity}, for a fixed $a \neq 0$,
\[
   \Pr\big[\wtH(A_\G a) < w^*\big] \;\le\; 2^{-(N-1)}\sum_{i<w^*}\binom{N}{i} \;\le\; w^* 2^{-(N-1)}\left(\frac{eN}{w^*}\right)^{w^*}.
\]
A union over the choices of $a$ with $\wtH(a) < w^*$ contributes another factor $w^*\left(\frac{eN}{w^*}\right)^{w^*}$, giving
\[
   \Pr\big[\G \text{ is not } w^*\text{-diagonal}\big] \;\le\; {w^*}^2\, 2^{-(N-1)}\left(\frac{eN}{w^*}\right)^{2w^*} \;\le\; 2^{-\delta' N}
\]
for some constant $\delta' > 0$.

\smallskip\noindent
\emph{Total.} $2^{-\delta N} + 2^{-\delta' N} \le 2^{-\delta(\varepsilon_1) N}$ for some constant $\delta(\varepsilon_1) > 0$.

\smallskip\noindent
\emph{Certificates.} By Definition~\ref{def:compatible}, failure of either property is the existence of the displayed vector. The verifier computes $A_\G a$ in time $O(N^2)$, forms $\Cl_\G(E) = b \oplus A_\G a$, tests $\wtS(\cdot) < w^*$, and tests membership in $C' \setminus \{0\}$ by checking $H'\Cl_\G(E) = 0$ and $\Cl_\G(E) \neq 0$, all in time $O(N^2)$.
\end{proof}

\section{Fine-grained hardness of the quantum minimum distance}
\label{sec:finegrained}
This section assembles the machinery of earlier sections into the fine-grained lower bounds. Throughout, $\varepsilon_1 \in (0,1)$ is an \emph{admissible} parameter: small enough that the failure exponent $\delta(\varepsilon_1) > 0$ of Lemma~\ref{lem:randomcompat} exists, i.e.\ $c_1(\varepsilon_1) < \frac{3}{4}$, where
\[
   c_1(\varepsilon_1) \;:=\; \varepsilon_1\log_2\left(\frac{3e}{\varepsilon_1}\right)
\]
is the exponent governing the number of low-weight witnesses at level $w=\varepsilon_1 N$. All constants introduced below are functions of $\varepsilon_1$ alone (and never of the instance), and the hypothesis constants (for e.g. $\varepsilon$ in SETH and $\beta$ in Conjecture~\ref{conj:BL}) are quantified before $\varepsilon_1$ is chosen.

\subsection{The reduction template}
\label{sec:template}

\begin{definition}[padded code]
\label{def:pad}
For any $[m, k]$ code $C$ (given by a parity check matrix $H$) and $N \ge m$, the code $\mathrm{pad}(C, N) := \{ (c, 0^{N-m}) : c \in C \}$ is an $[N, k]$ code with the same minimum distance, and a parity-check matrix for it is computable in time $O(N^2)$. Given a level $w$,
\[
   N \;:=\; \max\!\Bigg( 4k,\; m,\; \bigg\lceil\frac{w}{\varepsilon_1}\bigg\rceil \Bigg)
\]
yields $w \le \varepsilon_1 N$ and $k \le \frac{N}{4}$, with $N = O_{\varepsilon_1}(w + m)$.
\end{definition}

\begin{lemma}[randomized classical code to stabilizer code]
\label{lem:template}
Let $C \subseteq \F_2^m$ be a linear code of rank $\kappa$, and let $t, g \ge 0$ be integers. Set $w^* := t + g + 1$ and let $N$ be as in Definition~\ref{def:pad} with $w = w^*$. There is a randomized $\mathrm{poly}(N)$-time map producing an $[[N, \kappa]]$ stabilizer code $Q$ such that:
\begin{enumerate}[(i)]
\item \emph{(YES side)} if $\dist(C) \le t$ then $\qdist(Q) \le t$, with probability $1$ (that is, for every choice of the randomness).
\item \emph{(NO side)} if $\dist(C) > t + g$ then $\qdist(Q) > t + g$, with probability at least $1 - 2^{-\delta(\varepsilon_1)N}$.
\end{enumerate}
\end{lemma}

\begin{proof}
Pad $C$ to get $C'$ of length $N$ (Definition~\ref{def:pad}; rank and distance preserved), sample $\G$ uniformly in $\mathcal{G}_N$, and output $Q := Q_\G = \CWS(\G, C')$ with the presentation of Lemma~\ref{lem:threshold}(a).
Item (i) is Lemma~\ref{lem:threshold}(a) verbatim. For (ii), Lemma~\ref{lem:randomcompat} (applicable since $w^* \le \varepsilon_1 N$ and $\kappa \le \frac{N}{4}$) gives that $G$ is $(C', w^*)$-compatible and $w^*$-diagonal except with probability $2^{-\delta(\varepsilon_1) N}$, and on that event Lemma~\ref{lem:threshold}(b2) gives $\qdist(Q_G) \ge w^* = t+g+1 >  t + g$.
\end{proof}

The reduction errs only on the NO side, and only through the graph: the YES conclusion holds for every sampled $\G$.
This one-sidedness is what the deterministic selection rule of derandomization exploits, and hence, composing with a deterministic classical base yields $\mathsf{coRP}$-type rather than $\mathsf{BPP}$-type consequence there.

\subsection*{Exact hardness under SETH}
\label{sec:fgexact}

\begin{theorem}[SETH-hardness of exact \QMDP]
\label{thm:fg}
Assume randomized SETH. For every $\varepsilon > 0$, there is no randomized algorithm that, given a stabilizer code $Q$ on $N$ qubits with $\kappa$ logical qubits and an integer $T$, decides $\qdist(Q) \le T$ in time $2^{(1-\varepsilon)\kappa} \cdot \mathrm{poly}(N)$. 
\end{theorem}
\begin{proof}
Suppose such an algorithm $A$ exists. Assume $A$ has error at most $\frac{1}{4}$ otherwise amplify by repeating. By Theorem~\ref{thm:strong}(b), under randomized SETH there is no randomized algorithm deciding exact MDP instances $(C, t)$ of rank $n'$ and block length $m'$ in time $2^{(1-\varepsilon) n'} \cdot \mathrm{poly}(m')$.
Given such an instance of MDP, fix any admissible $\varepsilon_1$ and run Lemma~\ref{lem:template} with $g := 0$; this takes $\mathrm{poly}(m')$ time and produces $Q$ on $N = O_{\varepsilon_1}(t + m') = \mathrm{poly}(m')$ qubits with $\kappa = n'$ logical qubits. Query $A$ on $(Q, t)$ and return its answer.

\smallskip \noindent
If $\dist(C) \le t$: $\qdist(Q) \le t$ always, and $A$ answers YES.

\smallskip \noindent
If $\dist(C) \ge t + 1$: with probability $\ge 1 -
2^{-\delta(\varepsilon_1) N}$, $\qdist(Q) \ge t + 1> t$, and $A$ answers NO.

\smallskip \noindent
The composition is a randomized algorithm that errs only if the sampled graph is bad (with probability at most $\le 2^{-\delta(\varepsilon_1) N}$) or $A$ errs (with probability at most $\le \tfrac14$), so its total error is at most $\tfrac14 + 2^{-\delta(\varepsilon_1) N} < \tfrac13$ for large enough $N$.
Since $N = \mathrm{poly}(m')$, its running time is $2^{(1-\varepsilon)n'} \cdot \mathrm{poly}(N) = 2^{(1-\varepsilon)n'} \cdot \mathrm{poly}(m')$. 
This is a randomized algorithm deciding the exact MDP in that time, contradicting Theorem~\ref{thm:strong}(b).
\end{proof}

Note that the above reduction is one-sided: by Lemma~\ref{lem:threshold}(a) the YES conclusion $\qdist(Q) \le t$ holds for \emph{every} $\G$, so the sampling can only cause an error on NO instances. 
If $A$ is itself one-sided, so is the composition, matching the RUR structure of the classical base.

\subsection*{Gap hardness under the non-uniform Gap-ETH}
\label{sec:fggap}

\begin{theorem}[Gap-ETH-hardness with a linear additive gap]
\label{thm:fg-gap}
Assume non-uniform Gap-ETH. For every admissible $\varepsilon_1$ there are constants $\alpha = \alpha(\varepsilon_1) > 0$ and $c_\kappa = c_\kappa(\varepsilon_1) > 0$ such that no $2^{o (N)}$ time randomized algorithm solves the promise problem $\mathrm{GapAddQDist}_{\alpha N}$ restricted to stabilizer codes on $N$ qubits with $\kappa \ge c_\kappa N$ logical qubits.
In particular, the hardness holds for instances simultaneously linear in the block length, the number of logical qubits, and the additive gap.
\end{theorem}

\begin{proof}
By Theorem~\ref{imp:gapmdp} and Observation~\ref{prop:linprofile}, under non-uniform Gap-ETH there is a constant $\gamma_2 > 1$ and a family of $\gamma_2$-MDP instances $(C, t)$ with rank $\kappa = \Theta(n)$, block length $m = \Theta(n)$, and threshold $t = \Theta(m)$, on which no randomized $2^{o(n)}$-time algorithm works. Set $g := \lfloor (\gamma_2 - 1) t \rfloor = \Theta(m)$.
The NO promise takes the form $\dist(C) > \gamma_2 t \ge t + g$.
Now apply the map from Lemma~\ref{lem:template}.
The YES instance remains $\dist(C) \leq t$. 
Moreover, $w^* = t + g + 1 = \Theta(m)$, so $N = \Theta_{\varepsilon_1}(m) = \Theta(n)$ (linear block length); $\kappa = \Theta(n) = \Theta(N)$ (linear rank); and the quantum additive gap $g = \Theta(N)$ (linear additive gap).
Since $g$ and $\kappa$ are both $\Theta(N)$, the constants $\alpha := \liminf_{n} g/N$ and $c_\kappa := \liminf_{n} \kappa/N$ of the statement are positive and depend only on $\varepsilon_1$ and on the constants of Observation~\ref{prop:linprofile}.
A $2^{o(N)}$-time algorithm for the quantum promise problem, composed with the template, decides the classical MDP in randomized time $2^{o(N)} + \mathrm{poly}(n) = 2^{o(n)}$, contradicting the hypothesis (which, being non-uniform, also covers randomized algorithms).
\end{proof}

\section{Derandomization under a circuit hypothesis}
\label{sec:hypD}

The randomness in the reductions of Section~\ref{sec:finegrained} enters in exactly one place: the sampling of the graph $\G$. This section removes it, under a standard derandomization hypothesis.

We first replace the sampling by a deterministic polynomial-time \emph{list} of candidate graphs, at least one of which is \emph{good}. Because the reductions of Section~\ref{sec:finegrained} err only on NO instances, the YES conclusion of Lemma~\ref{lem:threshold}(a) holds for every graph and hence a list suffices, at the price of Karp reductions becoming nonadaptive Turing (truth-table) reductions (\S\ref{sec:detred}).
We then give the transfer principle (\S\ref{sec:transfer}): under \emph{circuit hypothesis}, block-length hardness of the classical problem, the open question of \cite{SV19}, in its original unrestricted form, forwards to block-length hardness of \QMDP.

Similar to the earlier sections, $\varepsilon_1$ is admissible in the sense of Section~\ref{sec:finegrained}, and $\delta(\varepsilon_1) > 0$ and $c_1(\varepsilon_1)$ are the constants of Lemma~\ref{lem:randomcompat}, throughout this section.
\subsection*{The derandomization objects}
\label{sec:derandprimer}

\emph{Nondeterministic circuits.} A nondeterministic circuit $D(x, y)$ has a designated witness input $y$; it accepts $x$ if and only if $\exists y$ with $D(x, y) = 1$. Size is the gate count of the circuit.
A co-nondeterministic circuit accepts $x$ if and only if every $y$ satisfies $D(x, y) = 1$; a set is co-nondeterministically accepted if its complement is nondeterministically accepted. 
Hence, a statement of the form \emph{$L$ requires nondeterministic circuits of size $s(n)$} gives a non-uniform lower bound: no size-$s(n)$ circuit family accepts exactly $L \cap \{0,1\}^n$.

\medskip \noindent 
\emph{Hitting-set generators vs.\ pseudorandom generators.} A PRG's output distribution is indistinguishable from uniform by every test in a class (two-sided).
A hitting-set generator guarantees only: every test in the class accepting at least half of all strings accepts some string on the list (one-sided).
HSGs are the weaker object and are all we need: our ``test'' is the set of good graphs, of density $\ge 1 - 2^{-\delta(\varepsilon_1) N}$ (Lemma~\ref{lem:randomcompat}); we never need a small error, only that the list cannot consist entirely of the exponentially rare bad graphs.

\medskip \noindent
\emph{Why a hitting set suffices.} Our reduction is one-sided in the graph: by Lemma~\ref{lem:threshold}(a) the YES bound $\qdist(Q_\G) \le \lambda$ holds for \emph{every} graph $\G$, good or bad. Consequently, the generator needs to only \emph{intersect} the set of good graphs, never estimate their measure, and the weaker hitting-set object, which follows from a weaker hardness assumption than a full PRG, is exactly the tool the reduction can use.

\subsection{The hypothesis and the imported generator}
\label{sec:hyp}

\begin{hypothesis}[circuit hypothesis]
\label{hyp:circuit}
There exist $\delta_D>0$ and a language $L\in\E=\DTIME(2^{O(n)})$ such that for all sufficiently large $n$, $L\cap\{0,1\}^n$ has no nondeterministic Boolean circuit of size $2^{\delta_Dn}$.
\end{hypothesis}

The following is not stated explicitly in \cite{MV05}, but follows from \cite[Cor.~1.8]{MV05} together with the observation that hardness against nondeterministic circuits implies hardness against SV-nondeterministic ones, and a padding argument. We give the proof for completeness.
\begin{fact}[hitting sets from \emph{circuit hypothesis}; {\cite[Cor.~1.8]{MV05}}]
\label{fact:hsg}
Assume \emph{circuit hypothesis}. For every constant $c \ge 2$ there is a deterministic algorithm that, on input $\ell$, runs in time $\poly(\ell)$ and outputs a list $x_1, \dots, x_r \in \{0,1\}^{\ell}$ with $r = \poly(\ell)$, such that the list intersects every set $T \subseteq \{0,1\}^{\ell}$ of density $\ge \tfrac12$ accepted by a co-nondeterministic circuit of size $\ell^{c}$.
\end{fact}
    \begin{proof}
Consider \cite[Cor.~1.8]{MV05}: for every constant $\tau_0 > 0$ there is a constant $\gamma > 0$ and a \emph{deterministic} polynomial-time procedure which, given the truth table of a Boolean function $f\colon \{0,1\}^{\mu} \to \{0,1\}$ of SV-nondeterministic circuit complexity at least $2^{\tau_0 \mu}$, outputs a hitting set in $\{0,1\}^{n}$ with threshold $\tfrac12$ for co-nondeterministic circuits of size $n$, where $n = \lceil 2^{\gamma \mu} \rceil$.

\medskip\noindent
Three things separate this from the statement we want in Fact~\ref{fact:hsg}: the hypothesis is about SV-nondeterministic rather than plain nondeterministic circuits; the procedure expects a truth table rather than a language; and it ties the string length and the circuit size to the same parameter $n$, whereas we need length $\ell$ against size $\ell^{c}$ for an arbitrary constant $c$. We address these one by one.

\medskip\noindent
Let $L \in \mathsf{E}$ and $\delta_D > 0$ be as in \emph{circuit hypothesis}. Since an SV-nondeterministic circuit of size $s$ yields a nondeterministic circuit of size $O(s)$ \cite[\S2]{MV05}, hardness against nondeterministic circuits implies hardness against SV-nondeterministic circuits; hence $L \cap \{0,1\}^{\mu}$ has SV-nondeterministic circuit complexity at least $2^{\tau_0 \mu}$ for all sufficiently large $\mu$ and some constant $\tau_0 = \tau_0(\delta_D) > 0$. Let $\gamma = \gamma(\tau_0)$ be the constant of \cite[Cor.~1.8]{MV05}; given $\ell$ and $c$, set $n := \ell^{c}$ and $\mu := \lceil \gamma^{-1}\log n\rceil = O(\log \ell)$. 
Since $L \in \mathsf{E}$, the truth table of $L \cap \{0,1\}^{\mu}$ ($2^{\mu} = \poly(\ell)$ entries, each computable in time $2^{O(\mu)} = \poly(\ell)$) is computable deterministically in time $\poly(\ell)$.\footnote{This is where $L \in \mathsf{E}$, rather than $L \in \mathsf{NE} \cap \mathsf{coNE}$, is used, and it is what makes the generator deterministic rather than single-valued nondeterministic \textup{(}cf.\ \cite[Cor.~3.4 vs.\ Cor.~3.2]{MV05}\textup{)}.} Using the procedure of \cite[Cor.~1.8]{MV05}, obtain a hitting set $H \subseteq \{0,1\}^{n'}$, $n' \ge n$, with threshold $\tfrac12$ for co-nondeterministic circuits of size $n'$. Truncate each string to its first $\ell$ bits. The procedure runs in time $\poly(\ell)$, and hence the resulting list has $r \le \poly(\ell)$ entries.

\medskip\noindent
It remains to show that truncation preserves the guarantee, and this follows from the padding argument used in \cite[Cor.~1.8, Cor.~3.2]{MV05} itself. Let $T \subseteq \{0,1\}^{\ell}$ have density $\ge \tfrac12$ and be accepted by a co-nondeterministic circuit $D$ of size $\ell^{c}$.
Lift it to
\[
   T' \;:=\; \{ x \in \{0,1\}^{n'} : x|_{[\ell]} \in T \}.
\]
The set $T'$ has the same density as $T$, and it is accepted by the circuit that runs $D$ on the first $\ell$ input bits and ignores the rest, which is co-nondeterministic of size $\ell^{c} = n \le n'$. Both hypotheses of the hitting guarantee at length $n'$ are therefore met, so some $x \in H$ lies in $T'$, that is, $x|_{[\ell]} \in T$. Since $T$ was arbitrary, the truncated list is as claimed.
\end{proof}

\subsection{Additive gap versions}
\label{sec:bases}

Both of our classical bases are additive-gap versions of the minimum-distance problem, differing only in the reduction type: the first is randomized with one-sided error, the second deterministic. 
\begin{fact}[randomized base; {\cite[Thm.~32]{DMS03}}]
\label{fact:randbase}
There is a constant $\tau > 0$ such that $\mathrm{GapAddDist}_\tau$ is NP-hard under polynomial-time RUR reductions: NO instances map to NO instances always; YES instances map to YES instances except probability $2^{-s}$ for a security parameter $s$, in time $\mathrm{poly}(s)$.
\end{fact}
\begin{fact}[deterministic base; {\cite[\S1.1]{GJS25}}]
\label{fact:detbase}
There is a constant $\tau > 0$ such that $\mathrm{GapAddDist}_\tau$ is NP-hard under deterministic polynomial-time Karp reductions.%
\footnote{Stated in \cite[\S1.1]{GJS25}, based on the methods of
\cite{BGLR25}.}
\end{fact}

The following lemma is the additive-gap analogue of Lemma~\ref{lem:template}: it converts a classical instance with additive gap $\tau m$ into a quantum instance with additive gap linear in the block length, at a constant loss depending only on $\varepsilon_1$ and $\tau$. Unlike the exact and multiplicative routes, it imposes no condition on the profile of the classical instance.
In particular, $t$ is not restricted by $\sqrt{m}$.
\begin{lemma}
\label{lem:addbase}
Let $(C \subseteq \F_2^m, t)$ be a $\mathrm{GapAddDist}_\tau$ instance with $m \ge 3$ and $t \le m$. Set
\[
   g \;:=\; \lfloor \tau m \rfloor \qquad \text{and} \qquad w^* \;:=\; t + g + 1 .
\]
Let $C'$ and $N$ be as in Definition~\ref{def:pad}, and let $Q_\G := \CWS(\G, C')$ on $N$ qubits with $\kappa = \mathrm{rank}(C)$ logical qubits. Then:
\begin{enumerate}[(i)]
\item for \emph{every} $\G$: a YES instance of $\mathrm{GapAddDist}_\tau$ gives $\qdist(Q_\G) \le t$ and
\item for every $\G$ that is $(C', w^*)$-compatible and $w^*$-diagonal: a NO instance of $\mathrm{GapAddDist}_\tau$ gives $\qdist(Q_\G) > t + \tau m$.
\end{enumerate}
Moreover the additive quantum gap satisfies $\tau m \ge \alpha(\varepsilon_1,\tau)\cdot N$ for
\[
   \alpha(\varepsilon_1,\tau) \;:=\; \frac{\tau\varepsilon_1}{2+\tau} .
\]
\end{lemma}
\begin{proof} Let $\lambda:=\dist(C)$.
(i) follows directly from Lemma~\ref{lem:threshold}(a) with $\lambda \le t$.

\smallskip \noindent
(ii) Since $\lambda$ is an integer and we are in the NO instance of  $\mathrm{GapAddDist}_\tau$, we get $\lambda \ge t + \lfloor \tau m \rfloor + 1 > t + g$.
Hence Lemma~\ref{lem:threshold}(b2) at gap parameter $g$ gives $\qdist(Q_\G) \ge w^* = t + g + 1 > t + \tau m$.

\smallskip \noindent
By Definition~\ref{def:pad}, $N = \max(4\kappa, m, \lceil w^*/\varepsilon_1 \rceil)$ with $w^* = t+g+1 \le m + \tau m + 1$, so
\[\Bigg\lceil \frac{w^*}{\varepsilon_1} \Bigg\rceil \le \frac{(1+\tau)m + 1 + \varepsilon_1}{\varepsilon_1} \le \frac{(2+\tau)m}{\varepsilon_1}.\]
Also $4\kappa \le 4m \le (2+\tau)m/\varepsilon_1$ since $\varepsilon_1 < \frac{1}{2}$.
And hence, $N \le (2+\tau)m/\varepsilon_1$ giving $\frac{\tau m}{N} \;\ge\; \frac{\tau\,\varepsilon_1}{2 + \tau}$.
\end{proof}

\begin{theorem}[Hardness for a linear additive gap]
\label{thm:nphard-rand}
There exist constants $\alpha_1, \alpha_2> 0$ such that 
\begin{enumerate}[(a)]
\item 
$\mathrm{GapAddQDist}_{\alpha_1 N} \notin \mathsf{BPP}$ unless $\mathsf{NP} \subseteq \mathsf{BPP}$,
\item $\mathrm{GapAddQDist}_{\alpha_2 N} \notin \mathsf{coRP}$ unless $\mathsf{NP} \subseteq \mathsf{coRP}$.
\end{enumerate}
\end{theorem}
\begin{proof}
    Given a $\mathrm{GapAddDist}_\tau$ instance $(C,t)$ with $C\subseteq\F_2^m$, sample $\G$ uniformly.
    Apply Lemma~\ref{lem:addbase}, and output $(Q_\G,t)$ with additive gap $\tau m\ge \alpha_1 N$, where $\alpha_1 := \tau\varepsilon_1/(2+\tau)$.
    By Lemma~\ref{lem:randomcompat}, $\G$ is $(C',w^*)$-compatible and $w^*$-diagonal except with probability $2^{-\delta(\varepsilon_1)N}=2^{-\Omega(N)}$. Conclusion~(i) of Lemma~\ref{lem:addbase} holds for every $\G$; conclusion~(ii) holds on that event.
    Thus, every YES instance of $\mathrm{GapAddDist}_\tau$ is mapped to a YES instance of $\mathrm{GapAddQDist}_{\alpha_1 N}$, for every $\G$; and every NO instance is mapped to a NO instance except with probability $2^{-\Omega(N)}$ over the sampling of $\G$. In other words, our construction is a UR-reduction from $\mathrm{GapAddDist}_\tau$ to $\mathrm{GapAddQDist}_{\alpha_1 N}$. Composing with the RUR-reduction of Fact~\ref{fact:randbase} gives a randomized reduction from an NP-complete language $\mathrm{L}$ to $\mathrm{GapAddQDist}_{\alpha_1 N}$ that is unfaithful on either side with probability $2^{-\Omega(N)}$. Hence a polynomial-time randomized algorithm for $\mathrm{GapAddQDist}_{\alpha_1 N}$ with error at most $\tfrac14$ decides $\mathrm{L}$ with total error below $\tfrac13$, so $\mathsf{NP}\subseteq\mathsf{BPP}$.

    \smallskip \noindent
    (b) follows by the same argument, with the RUR-reduction of Fact~\ref{fact:randbase} replaced by the deterministic Karp reduction of Fact~\ref{fact:detbase} (with a constant $\alpha_2$ replacing $\alpha_1$). Being faithful on both sides, that base preserves the one-sidedness of our construction: the composition is again a UR-reduction from $\mathrm{L}$, faithful on YES instances and unfaithful on NO instances with probability $2^{-\Omega(N)}$. Hence if $\mathrm{GapAddQDist}_{\alpha_2 N}\in\mathsf{coRP}$ then $\mathrm{L}\in\mathsf{coRP}$ resulting in $\mathsf{NP}\subseteq\mathsf{coRP}$.
\end{proof}

The conclusion of \Cref{thm:nphard-rand}(b) can be stated more strongly.
Suppose $\NP\subseteq\coRP$ and let $L\in\coRP$.
Then $\bar L\in\RP\subseteq\NP\subseteq\coRP$, so $\bar L\in\coRP$ and hence $L\in\RP$.
Thus $\coRP\subseteq\RP$; and $\RP\subseteq\NP\subseteq\coRP$ gives the reverse inclusion, so $\RP=\coRP$ and $\ZPP=\RP\cap\coRP=\RP$.
Finally $\NP\subseteq\coRP=\RP\subseteq\NP$, so
\[
  \NP=\RP=\coRP=\ZPP.
\]
This is noted in \cite[\S2.2]{DMS03}, who obtain only $\NP=\RP$ because their reduction is faithful on \NO~instances where ours is faithful on \YES.

\begin{lemma}[list construction]
\label{lem:listB}
Assume \emph{circuit hypothesis}. There is a deterministic polynomial-time algorithm that, given an $[N, k']$ code $C'$ (by parity check matrix $H'$) with $k' \le \frac{N}{4}$ and a level $w^* \le \varepsilon_1 N$, outputs graphs $\G_1, \dots, \G_r \in \mathcal{G}_N$ with $r = \poly(N)$, such that at least one $\G_i$ is $(C', w^*)$-compatible and $w^*$-diagonal.
\end{lemma}

\begin{proof}
Identify $\mathcal{G}_N$ with $\{0,1\}^{\ell}$ for $\ell := \binom{N}{2}$, via the entries of the adjacency matrix above the diagonal, and let $\mathcal{B} \subseteq \{0,1\}^{\ell}$ be the set of \emph{bad} graphs for $C'$, that is,
\[
   \mathcal{B} \;:=\; \big\{\, \G \in \mathcal{G}_N \;:\; \G \text{ is not } (C', w^*)\text{-compatible, or not } w^*\text{-diagonal} \,\big\}.
\]
By the witness clause of Lemma~\ref{lem:randomcompat}, $\G \in \mathcal{B}$ if and only if there exists a witness $(a \mid b) \in \F_2^{2N}$, verifying failure in time $O(N^2)$ given $A_\G$ and $H'$. Hardwiring the parity check of $C'$, the verifier becomes a Boolean circuit $D_{C'}\big(\G, (a \mid b)\big)$ of size $\ell^{c_0}$ for an absolute constant $c_0 \ge 2$. Since $\G \in \mathcal{B}$ exactly when a valid witness exists, $D_{C'}$ is a nondeterministic circuit accepting $\mathcal{B}$, and hence $\overline{\mathcal{B}}$ is accepted by a co-nondeterministic circuit of the same size.

Moreover, by Lemma~\ref{lem:randomcompat} a uniformly random $\G$ lies in $\mathcal{B}$ with probability at most $2^{-\delta(\varepsilon_1)N}$, so $\overline{\mathcal{B}}$ has density at least $\tfrac12$ for large enough $N$. Running the generator of Fact~\ref{fact:hsg} at length $\ell$ with parameter $c_0$ therefore yields a list intersecting $\overline{\mathcal{B}}$, that is, containing at least one good graph.
\end{proof}

\subsection{Deterministic reductions}
\label{sec:detred}

\begin{definition}[promise-respecting truth-table reductions]
\label{def:ttred}
A promise problem is a pair $(\Pi_Y, \Pi_N)$ of disjoint sets. An oracle solves it if it answers YES on $\Pi_Y$ and NO on $\Pi_N$ (answers elsewhere arbitrary).
A deterministic polynomial-time nonadaptive Turing (truth-table) reduction from $L$ to $(\Pi_Y, \Pi_N)$ computes, from the instance $x$ alone, a list of queries together with a Boolean selection rule, and must decide $x$ correctly for every oracle solving the promise problem.
\end{definition}

The reductions below are nonadaptive but not Karp reductions. The queries are the instances built from a list of candidate graphs, and the selection rule decides based on their answers. 

\begin{theorem}[deterministic reductions under \emph{circuit hypothesis}]
\label{thm:detred}
Assume \emph{circuit hypothesis} and fix an admissible $\varepsilon_1$. Throughout, $N$ denotes the number of physical qubits of the stabilizer code and $\kappa$ its number of logical qubits.
\begin{enumerate}[(i)]
\item There is a constant $\alpha = \alpha(\varepsilon_1) > 0$ such that $\mathrm{GapAddQDist}_{\alpha N}$, the problem of deciding whether $\qdist(Q) \le T$ or $\qdist(Q) > T + \alpha N$ for a stabilizer code $Q$ on $N$ qubits and a threshold $T$, is NP-hard under deterministic polynomial-time truth-table reductions.
\item Assuming (deterministic) SETH, for any $\varepsilon > 0$ there is no deterministic $2^{(1-\varepsilon)\kappa}\cdot\poly(N)$-time algorithm for exact $\mathrm{\QMDP}$.
\item Assuming non-uniform Gap-ETH, the linear-gap hardness of Theorem~\ref{thm:fg-gap}, on instances with $\kappa = \Theta(N)$, holds with the reduction made deterministic.
\end{enumerate}
\end{theorem}
\begin{proof} \emph{Reduction:}
Let $(C,t)$ be a classical minimum-distance instance with additive gap $g \ge 0$, the exact case being $g = 0$.
Set $w^* := t+g+1$ and pad $C$ per Definition~\ref{def:pad} to get $C'$.
Using Lemma~\ref{lem:listB} on $C'$ at level $w^*$, obtain the list of graphs $\G_1, \dots, \G_r$.
Query the oracle\footnote{$\mathrm{\QMDP}$ when $g = 0$, and  $\mathrm{GapAddQDist}_{g}$ when $g > 0$, the gap being that of Lemma~\ref{lem:addbase}.} with $(Q_{\G_i},t)$ for each $i$. 
Output \textrm{NO} if any of the queries return \textrm{NO}, otherwise output \textrm{YES}. 

\medskip\noindent 
\emph{Correctness:} Suppose $(C,t)$ is a \textrm{YES} instance, that is $\dist(C):= \lambda \leq  t$. 
Then by Lemma~\ref{lem:threshold}(a), for all $i$, we have, $\qdist(Q_{\G_i}) \le \lambda \le t$.
Every query lies on the YES side of the promise, and hence every oracle call returns YES. 
So the output of the entire reduction is YES.
On the other hand, suppose $(C,t)$ is a NO instance, that is,  $\lambda > t +g$.
By Lemma~\ref{lem:listB} some $i^*$ is \emph{good}, and by Lemma~\ref{lem:threshold}(b2), $\qdist(Q_{\G_{i^*}}) \ge t + g + 1 > t + g$, and thus query $i^*$ returns NO on it, and hence we get the output NO. Note that queries $i \neq i^*$ may violate the promise (bad graphs can lower the distance into the gap) and may be answered arbitrarily, which is harmless. Thus, we get a valid reduction that is deterministic, polynomial-time, and nonadaptive.

\medskip\noindent 
For (i), compose the reduction above with the Karp base of Fact~\ref{fact:detbase}.
This gives a deterministic truth-table reduction from SAT, at quantum gap $\alpha N$ by Lemma~\ref{lem:addbase}.

\smallskip\noindent 
For (ii), consider $g = 0$.
The queries are then instances of $\mathrm{\QMDP}$, a total decision problem, so a deterministic algorithm for $\mathrm{\QMDP}$ can play the role of the oracle.
Given such an algorithm running in time $2^{(1-\varepsilon)\kappa} \cdot \poly(N)$, the reduction makes $r = \poly(N)$ calls, each at $\kappa = n'$, and so decides exact MDP in time $2^{(1-\varepsilon)n'} \cdot \poly(m')$, contradicting Theorem~\ref{thm:strong}(a) under SETH.

\smallskip\noindent 
For (iii), replace the sampling of $\G$ in the proof of Theorem~\ref{thm:fg-gap} by the list of Lemma~\ref{lem:listB} and the selection rule; the conclusion regarding the lower bound remains unchanged since non-uniform Gap-ETH covers randomized algorithms as well.
\end{proof}

\subsection{The transfer principle}
\label{sec:transfer}

Our fine-grained results are parametrized by $\kappa$, the number of logical qubits, rather than by the block length. This is inherited from the classical input: the SETH-hardness of \cite{SV19} is stated in the rank, and on their instances the block length is polynomially larger, so $\kappa \le \frac {N}{4}$ on the codes we produce. Block-length hardness for the classical problem is their own open question \cite[\S1.2]{SV19}. We now show that, under \emph{circuit hypothesis}, our reduction forwards any positive answer to that question, in its original, unrestricted form, to quantum block-length hardness, with an explicit constant-factor loss in the exponent.

\begin{conjecture}[{\cite{SV19}}]
\label{conj:BL}
There is a constant $\beta > 0$ such that no algorithm decides exact $\mathrm{MDP}$ on codes of block length $m$ in time $2^{\beta m}$.
\end{conjecture}
\begin{theorem}[transfer, under \emph{circuit hypothesis}]
\label{thm:transfer}
Assume \emph{circuit hypothesis} and Conjecture~\ref{conj:BL}. Then for every admissible $\varepsilon_1$ and every $c < \beta \varepsilon_1$ no deterministic algorithm decides $\mathrm{\QMDP}$ on stabilizer codes of block length $N$ in time $2^{cN}$.
\end{theorem}
\begin{proof}
We will give a reduction from $\mathrm{MDP}$ to $\poly(N)$ instance of $\mathrm{\QMDP}$.

\smallskip\noindent 
Suppose an  algorithm $A$ decides exact $\mathrm{\QMDP}$ in time $2^{c\,N}$ for some $c < \beta\varepsilon_1$.
Define the margin parameter $\eta := 1 - \frac{c}{\beta\varepsilon_1}$.
Given an instance $(C, t)$ of $\mathrm{MDP}$, we may assume $t \le m$ (otherwise the answer is deterministically YES, as $\lambda \le m$ always) and $\kappa \le m$.
Set $w^* := t+1 \le m+1$ and pad via Definition~\ref{def:pad}:
\[
   N \;=\; \max\!\Big( 4\kappa,\; m,\; \bigg\lceil \frac{w^*}{\varepsilon_1} \bigg\rceil \Big)   \;\le\; \frac{m}{\varepsilon_1}\,\big(1 + o(1)\big),
\]
since $4\kappa \le 4m \le \frac{m}{\varepsilon_1}$ for admissible $\varepsilon_1$ and $\bigg\lceil \frac{t+1}{\varepsilon_1} \bigg\rceil \le \frac{m+1}{\varepsilon_1} + 1$.

\smallskip \noindent 
Use Lemma~\ref{lem:listB} to obtain, in time $\mathrm{poly}(N)$, candidates $\G_1, \dots, \G_r$ with $r = \mathrm{poly}(N)$ containing at least one good $\G_i$.
Query for each $Q_{\G_i}$ at threshold $t$ and output NO if some query answers NO.
Otherwise output YES.
Correctness follows by the same argument as in Theorem~\ref{thm:detred} at gap parameter $g = 0$, with the total decision problem in place of the promise oracle: on a YES instance ($\lambda \le t$), every candidate, good or bad  satisfies $\qdist(Q_{\G_i}) \le \lambda \le t$ by Lemma~\ref{lem:threshold}(a), so every query answers YES; on a NO instance ($\lambda \ge t+1$), the good candidate $i^*$ satisfies $\qdist(Q_{\G_{i^*}}) > t$ by Lemma~\ref{lem:threshold}(b2), supplying a NO vote.

\smallskip \noindent 
The total running time is $\mathrm{poly}(m) + r \cdot 2^{c\,N}
\cdot \mathrm{poly}(N) = 2^{c\,N + O(\log m)}$, and
\[
   c\,N \;\le\; c \cdot \frac{m}{\varepsilon_1}(1+o(1))\, 
   \;=\; \beta\,(1-\eta)(1+o(1))\, m
   \;\le\; \beta\Big(1 - \frac{\eta}{2}\Big) m
\]
for all sufficiently large $m$. The composed deterministic algorithm therefore decides the family in time $2^{\beta m - \Omega(m)} < 2^{\beta m}$, contradicting Conjecture~\ref{conj:BL}. 
Hence, no algorithm decides exact $\mathrm{\QMDP}$ in time $2^{c\,N}$.
\end{proof}

\newpage
\textcolor{red}{ 
}

\bibliography{bib}
\bibliographystyle{alpha}

\appendix

\section{Proofs of standard facts}\label{app:standard}

\begin{proof}[Proof of Fact~\ref{fact:commutation}]
The single-qubit relation $ZX = -XZ$ gives the first identity coordinatewise; the second follows by moving $Z(b)$ past $X(a')$. Nondegeneracy: if  $e = (a \mid b) \neq 0$, pick $i$ with $(a_i, b_i) \neq (0,0)$ and take $f$ supported on qubit $i$ with $(a'_i, b'_i)$ chosen so that $a_i b'_i \oplus b_i a'_i = 1$ (possible for each of the three nonzero values of $(a_i, b_i)$).
\end{proof}

\begin{proof}[Proof of Fact~\ref{fact:dim-sign}]
(a) If $S, S' \in \cS$ have $\sigma(S) = \sigma(S')$ then $S S'^{-1} \in \{\pm I, \pm i I\} \cap \cS$; since elements of $\cS$ square to $\pm$-phases of $I$ and $\cS$ is closed under products, $S S'^{-1} \in \{\pm I\}$, and $-I \notin \cS$ forces $S = S'$. Injectivity plus independence of the generators' images gives $|\cS| = 2^r$. Every non-identity $S \in \cS$ has $\sigma(S) \neq 0$, hence $\tr S = 0$ (the trace of $X(a) Z(b)$ vanishes unless $a = b = 0$, factorwise).
Thus $\tr \Pi = 2^{-r} \tr I = 2^{N - r}$; and $\Pi$ is a projector (a product of commuting projectors), so $\dim Q = \tr \Pi$.

(b) The signed generators commute (commutation depends only on $\sigma$), are projectively independent, and square to $I$. It remains to check that $-I$ is not in the generated group: any element with trivial symplectic part is a product $\prod_j (s_j P_j)^{x_j}$ with $\sum_j x_j v_j = 0$, hence $x = 0$ by independence, i.e.\ the empty product $+I$. The trace computation of (a) applies verbatim.
\end{proof}

\begin{proof}[Proof of Fact~\ref{fact:normalizer}]
Conjugation of a Pauli by a Pauli gives $P S P^\dagger = \pm S$ (Fact~\ref{fact:commutation}). If $P \in N(\cS)$ and $P S P^\dagger = -S$ for some $S \in \cS$, then $-S \in \cS$, and since also $S \in \cS$ we get $-I = (-S) S^{-1} \in \cS$, a contradiction. So conjugation fixes each element: $N(\cS) \subseteq C(\cS)$; the reverse inclusion is trivial. The symplectic characterization is Fact~\ref{fact:commutation}.
\end{proof}

\begin{proof}[Proof of Fact~\ref{fact:trichotomy}]
(1) Pick $T \in \cS$ with $T E = -E T$ (it exists by Fact~\ref{fact:normalizer} and Fact~\ref{fact:commutation}). Then $\Pi E \Pi = \Pi T E \Pi = -\Pi E T \Pi = -\Pi E \Pi$, using $T \Pi = \Pi T = \Pi$. Hence $\Pi E \Pi = 0$.

(2) By injectivity of $\sigma$ on $\cS$ (Fact~\ref{fact:dim-sign}(a)) there is $S \in \cS$ with $\sigma(S) = e$, and then $E S^{-1}$ has trivial symplectic part, i.e.\ $E = i^\gamma S$. Then $\Pi E \Pi = i^\gamma \Pi S \Pi = i^\gamma \Pi$.

(3) $E$ commutes with $\cS$ (Fact~\ref{fact:normalizer}), so $E \Pi = \Pi E$ and $E$ maps $Q$ to $Q$ unitarily; $\Pi E \Pi = E \Pi$.
Suppose toward a contradiction that $E \Pi = c \, \Pi$. Since $e \notin \bS = (\bS^\perp)^\perp$, the linear functional $f \mapsto \omega(e, f)$ is not identically zero on $\bS^\perp$; pick $f \in \bS^\perp$ with $\omega(e, f) = 1$ and a Pauli $F$ with $\sigma(F) = f$. Then $F \in C(\cS)$, so $F \Pi = \Pi F$, and $F^\dagger E F = -E$. Applying $F^\dagger (\cdot) F$ to $E \Pi = c \Pi$ gives
\[
- E \Pi \;=\; F^\dagger E F \, \Pi \;=\; F^\dagger E \Pi F \;=\; c \, F^\dagger \Pi F \;=\; c \, \Pi ,
\]
hence $E \Pi = -c \Pi$, so $c = 0$ and $E \Pi = 0$; impossible, as $E$ is unitary and $\Pi \neq 0$. Therefore $E \Pi$ is not a scalar multiple of $\Pi$: $E$ is undetectable and its logical action is not proportional to the identity.
\end{proof}

\section{Proofs of stronger forms of \cite{SV19}}
\begin{proof}[Proof of \Cref{thm:strong}]
We will first prove ($a$).
Suppose, for contradiction, that such an algorithm $A$ exists, so that on a rank$=n'$, block-length$=m'$ instance it runs in time
\begin{equation}\label{eq:A}
   T_A(n',m') \;=\; 2^{(1-\varepsilon)\,n'}\cdot \mathrm{poly}(m').
\end{equation}
Fix $\varepsilon>0$ and let $k=k(\varepsilon)$ be the constant supplied by Hypothesis~\ref{imp:sat}; it suffices to build a $k$-SAT algorithm running in time $2^{(1-\varepsilon'')n}$ for some fixed $\varepsilon''>0$, since this contradicts \textup{SETH}.
The reduction chains the two classical imports
\[
   k\text{-SAT}
   \;\xrightarrow{\ \text{Corollary~\ref{imp:ncp}}\ }\;
   \textup{NCP}
   \;\xrightarrow{\ \text{Theorem~\ref{imp:turing}}\ }\;
   \text{exact \textup{MDP}},
\]
and then invokes $A$; the entire argument reduces to bounding a single exponent, which we now do.
 
\medskip
\noindent\emph{From $k$-SAT to NCP.} By Corollary~\ref{imp:ncp}, the $k$-SAT instance on $n$ variables maps in deterministic $\mathrm{poly}(n)$ time to an \textup{NCP} instance of rank $n_{\mathrm{NCP}}\le n$ and block length $O(m)$. The rank bound is an inequality, but this only helps: a strictly smaller rank can only decrease the exponent $(1-\varepsilon)n'$ produced downstream, and hence only strengthens the final contradiction. We therefore set $n_{\mathrm{NCP}}=n$ in the running-time bookkeeping, with the understanding that every subsequent inequality is preserved when $n_{\mathrm{NCP}}<n$.
 
\medskip
\noindent\emph{From NCP to exact MDP.}
Apply the deterministic Turing reduction of Theorem~\ref{imp:turing} to the rank $n$ \textup{NCP} instance. It produces $Q=2^{3n/4}$ exact-\textup{MDP} instances, each of rank
\begin{equation}\label{eq:rank}
   n' \;=\; \Big\lceil \tfrac{(1+\varepsilon')n}{4}\Big\rceil + 1
   \;\le\; \tfrac{(1+\varepsilon')n}{4} + 2,
\end{equation}
where $\varepsilon'\in(0,\tfrac12)$ is a rank-blow-up parameter we are free to choose.
The reduction has two exponential costs from two steps, which we account for separately: the deterministic construction of the augmented gadget and the enumeration of the $2^{3n/4}$ instances, running in time $2^{3(1+\varepsilon')n/4}\cdot\mathrm{poly(n,m)}$; and the $Q=2^{3n/4}$ calls to $A$.
Crucially, these two phases run \emph{sequentially}, so the total time is their \emph{sum}.
 
\medskip
\noindent\emph{Cost of the oracle calls.}
One call to $A$ costs $2^{(1-\varepsilon)n'}\cdot \mathrm{poly}(m')$ by \eqref{eq:A}, so $2^{3n/4}$ calls contribute
\[
   2^{3n/4}\cdot 2^{(1-\varepsilon)n'}\cdot\mathrm{poly}(m')
   \;=\; 2^{\,\frac{3n}{4}+(1-\varepsilon)n'}\cdot\mathrm{poly}(m').
\]
Substituting the bound \eqref{eq:rank} on $n'$,
\[
   \tfrac{3n}{4}+(1-\varepsilon)n'
   \;\le\;
   \Big[\tfrac34+\tfrac{(1-\varepsilon)(1+\varepsilon')}{4}\Big]n
   \;+\; 2(1-\varepsilon),
\]
so the oracle phase runs in $2^{e_{\mathrm{oracle}}\,n}\cdot\mathrm{poly}(m')$ with $e_{\mathrm{oracle}}=\tfrac34+\tfrac{(1-\varepsilon)(1+\varepsilon')}{4}$, the additive $2(1-\varepsilon)=O(1)$ absorbed into a constant factor. 
\medskip

\noindent\emph{Combining the phases.}
The construction phase has exponent $e_{\mathrm{constr}}=\tfrac{3(1+\varepsilon')}{4}$ directly from Theorem~\ref{imp:turing}. Since the two phases run sequentially, a sum of two exponentials is bounded by twice the larger, so the derived $k$-SAT algorithm runs in time $2^{E n}\cdot\mathrm{poly}(m')$ with
\begin{equation}\label{eq:maxE}
   E \;=\; \max\!\left(  \frac{3(1+\varepsilon')}{4},\;\;  \frac34+\frac{(1-\varepsilon)(1+\varepsilon')}{4} \right).
\end{equation}

\medskip
\noindent\emph{Absorbing the block-length factor.}
By Observation~\ref{lem:blocklen}, the hard instances have block length $m'=n^{\Theta(1/\varepsilon')}$, which for constant $\varepsilon'$ is $\mathrm{poly}(n)$. Hence $\mathrm{poly}(m')=\mathrm{poly}(n)=2^{o(n)}$, and the polynomial factor carried since \eqref{eq:A} does not affect the exponent $E$. This is precisely where the strengthening from the bare bound to the $\mathrm{poly}(m')$ bound is obtained, and it is why the stronger form costs nothing.
 
\medskip
\noindent\emph{Forcing $E<1$.}
It remains to choose $\varepsilon'$ so that both arguments of \eqref{eq:maxE} are below $1$. The first satisfies $\tfrac{3(1+\varepsilon')}{4}<1$ iff $\varepsilon'<\tfrac13$. 
The second satisfies $\tfrac34+\tfrac{(1-\varepsilon)(1+\varepsilon')}{4}<1$ iff $(1-\varepsilon)(1+\varepsilon')<1$, i.e.\ $\varepsilon'<\varepsilon/(1-\varepsilon)$.
Both hold whenever
\begin{equation}\label{eq:cond}
   \varepsilon' \;<\; \min\!\Big(\varepsilon,\ \tfrac13\Big);
\end{equation}
indeed $\varepsilon<\tfrac12$ (Hypothesis~\ref{imp:sat}) gives $\varepsilon\le\varepsilon/(1-\varepsilon)$, so $\varepsilon'<\varepsilon$ implies the second condition. Since $\varepsilon>0$ is fixed and $\varepsilon'$ is ours to choose, \eqref{eq:cond} is satisfiable, and with such a choice $E=1-\varepsilon''$ for a fixed $\varepsilon''>0$.
 
\medskip
\noindent Combining all this, the derived $k$-SAT algorithm runs in
time
\[
   2^{En}\cdot\mathrm{poly}(m')
   \;=\; 2^{(1-\varepsilon'')n}\cdot 2^{o(n)}
   \;=\; 2^{(1-\varepsilon''+o(1))n},
\]
contradicting Hypothesis~\ref{imp:sat}. Therefore, no algorithm of the form given in $(a)$ exists.

\medskip
\noindent  Now we will move to part $(b)$.

\medskip
\noindent
\medskip
The randomized route is shorter, because the reduction Theorem~\ref{imp:rs} \textup{(}\cite[Cor.~5.7]{SV19}\textup{)} is polynomial-time and produces a \emph{single} $\mathrm{MDP}$ instance, so there is no enumeration phase and no maximum to take.

Suppose a randomized $A$ solves $\mathrm{MDP}_{n',m'}$ in time $2^{(1-\varepsilon)n'}\cdot\mathrm{poly}(m')$. 
Compose Corollary~\ref{imp:ncp} ($k$-SAT $\to$ NCP, deterministic, rank $\le n$) with Theorem~\ref{imp:rs} (NCP $\to$ MDP, randomized polynomial-time, rank $\le(1+\varepsilon')n$), and run $A$ once.
The derived $k$-SAT algorithm runs in
\[
   2^{(1-\varepsilon)(1+\varepsilon')\,n}\cdot\mathrm{poly}(m')
   \;=\; 2^{(1-\varepsilon)(1+\varepsilon')\,n}\cdot 2^{o(n)},
\]
using $n'\le(1+\varepsilon')n$ and $\mathrm{poly}(m')\le 2^{o(n)}$ (Observation~\ref{lem:blocklen}). 
Again, $\varepsilon'<\frac{\varepsilon}{1-\varepsilon}$ makes the exponent $1-\varepsilon''$ for a fixed $\varepsilon''>0$. 
Since the reduction is randomized, so is the composed $k$-SAT algorithm; hence the hypothesis it contradicts is the randomized form of \textup{SETH}.
\end{proof}

\begin{proof}[Proof of Observation~\ref{prop:linprofile}]
We follow the four steps of the gap reduction of \cite[\S5.2]{SV19} and track the three parameters through each. Throughout, $q = 2$ and $k = 3$.
 
\medskip
\noindent\emph{Step 1: the \textup{Gap-3-SAT} instance is sparse and has no unused variables.}
By \cite[Thm.~3.4]{SV19}, there are constants $C$ and $s \in (0,1)$ such that Gap-ETH may be assumed to hold already for $(s,1)$-\textup{Gap-3-SAT} instances on $n$ variables in which each
variable appears in at most $C$ clauses. We may further assume that every variable occurs in at least one clause: deleting unused variables leaves the clause set and the promise unchanged and only decreases the variable count, so a $2^{o(n)}$-time algorithm for instances without unused variables yields one for the general family, and the hardness of \cite[Thm.~3.4]{SV19} transfers. 
Let $m_{\mathrm{SAT}}$ be the number of clauses and count literal occurrences: each clause has three literals and each variable occurs between $1$ and $C$ times, so
\[
   n/3 \;\le\; m_{\mathrm{SAT}} \;\le\; Cn/3,
   \qquad\text{i.e.}\qquad
   m_{\mathrm{SAT}} = \Theta(n).
\]
 
\medskip
\noindent\emph{Step 2: the gap-\textup{NCP} instance.}
Apply the gap-preserving reduction of Theorem~\ref{imp:ncp_original} with $q = 2$, $k = 3$, and $c = 1$. By Observation~\ref{lem:fullrank}
(applicable by Step~1), the output instance has
\[
   \text{rank } = n, \qquad m_{\mathrm{NCP}} \;=\; (2^{3}-1)\,m_{\mathrm{SAT}} \;=\; 7\,m_{\mathrm{SAT}} \;=\; \Theta(n),
\]
and it is a $\gamma$-\textup{NCP} instance with threshold
\[
   t_{\mathrm{NCP}} \;=\; 7\,m_{\mathrm{SAT}} - 4\,m_{\mathrm{SAT}}
     \;=\; 3\,m_{\mathrm{SAT}} \;=\; \Theta(n)
\]
and gap
\[
   \gamma \;=\; \frac{1 - s/2 - 1/8}{1 - 1/2 - 1/8}
   \;=\; \frac{7 - 4s}{3} \;>\; 1,
\]
a constant since $s < 1$ is a constant. Explicitly, YES instances satisfy $\dist(y,C) \le 3m_{\mathrm{SAT}}$ and NO instances $\dist(y,C) > (7-4s)\,m_{\mathrm{SAT}}$. 
Since $t_{\mathrm{NCP}} = 3m_{\mathrm{SAT}} \ge n$ and the rank equals $n$, the hypothesis of \cite[Cor.~5.10]{SV19} \textup{(}their ``$d \ge \varepsilon n$'', i.e.\ our $t_{\mathrm{NCP}} \ge \varepsilon\cdot\mathrm{rank}$\textup{)} is satisfied.
 
\medskip
\noindent\emph{Step 3: the kissing-code gadget is linear in every parameter.}
By \cite[Cor.~5.9]{SV19} there are constants $\varepsilon^\dagger > 0$ and $\gamma^\dagger \in (1,2)$ such that for \emph{every} sufficiently large block length $m^\dagger$ there exists an $M$-locally dense triple $(C^\dagger, y^\dagger, r^\dagger)$ with $M \ge 2^{\varepsilon^\dagger m^\dagger}$ and $\dist(C^\dagger) \ge \gamma^\dagger r^\dagger$. In the proof of \cite[Cor.~5.10]{SV19} one sets $\alpha := 1 + 1/\varepsilon^\dagger$ and
\[
   m^\dagger \;:=\; \lfloor \alpha n \rfloor \;=\; \Theta(n);
\]
this choice is legitimate precisely because the triples of \cite[Cor.~5.9]{SV19} exist at every sufficiently large block length.
(The density requirement $M \ge 10\cdot 2^{n}$ of \cite[Cor.~5.3]{SV19} is also met: $\varepsilon^\dagger m^\dagger \ge \varepsilon^\dagger(\alpha n - 1) = (1+\varepsilon^\dagger)n - \varepsilon^\dagger \ge n + 4$ for $n$ sufficiently large, so $M \ge 16 \cdot 2^{n}$.)
 
Write $n^\dagger$ for the rank of $C^\dagger$. We state the linearity of all three gadget parameters.
 
\smallskip
\noindent\emph{Gadget rank.} Trivially $n^\dagger \le m^\dagger$.
Conversely, the $M$ vectors $z^\dagger$ counted in the definition of an $M$-locally dense triple \textup{(}\cite[Def.~5.1]{SV19}\textup{)} are distinct elements of $\F_2^{n^\dagger}$, so $2^{n^\dagger} \ge M \ge 2^{\varepsilon^\dagger m^\dagger}$.
Hence
\[
   \varepsilon^\dagger m^\dagger \;\le\; n^\dagger \;\le\; m^\dagger,
   \qquad n^\dagger = \Theta(n).
\]
 
\smallskip
\noindent\emph{Gadget radius.} The statement of \cite[Cor.~5.9]{SV19} bounds $r^\dagger$ only from above, via $r^\dagger \le \dist(C^\dagger)/\gamma^\dagger \le m^\dagger/\gamma^\dagger$. The matching lower bound can also be inferred from the same construction: there one sets $r^\dagger := \lfloor \dist(\widehat{C})/\gamma^\dagger \rfloor$ \textup{(}their $d^\dagger$\textup{)}, where $\widehat{C}$ is the Ashikhmin--Barg--Vl\u{a}du\c{t} code of \cite[Thm.~5.8]{SV19} \textup{(}following \cite{ABV01}\textup{)}, which satisfies $\dist(\widehat{C}) \ge \delta^\dagger m^\dagger$ for a constant $\delta^\dagger > 0$. 
Hence
\[
   \frac{\delta^\dagger}{\gamma^\dagger}\,m^\dagger - 1
   \;\le\; r^\dagger \;\le\; \frac{m^\dagger}{\gamma^\dagger},
   \qquad r^\dagger = \Theta(m^\dagger) = \Theta(n).
\]
 
\medskip
\noindent\emph{Step 4: the repetition count is constant.}
The reduction of \cite[Cor.~5.3]{SV19} stacks $k_{\mathrm{rep}}$ copies of the gadget in order to force a constant output gap $\gamma_2 > 1$; with $\delta := \gamma^\dagger - 1 > 0$, the proof of \cite[Cor.~5.10]{SV19} takes
\[
   k_{\mathrm{rep}} \;:=\; \Big\lceil
   \tfrac{2\,t_{\mathrm{NCP}}}{\delta\, r^\dagger} \Big\rceil
   \qquad\textup{(}\text{in the notation of \cite{SV19}: }
   \lceil 2d/(\delta d^\dagger)\rceil\textup{)}.
\]
By Steps 2 and 3, $t_{\mathrm{NCP}} = \Theta(n)$ and $r^\dagger = \Theta(n)$, and $\delta$ is a constant; hence $1 \le k_{\mathrm{rep}} = O(1)$, i.e.\ $k_{\mathrm{rep}} = \Theta(1)$.
This is the step at which linearity is preserved: a constant number of copies of a linear-sized gadget is still linear. Note that the constancy of $k_{\mathrm{rep}}$ uses the \emph{lower} bound $r^\dagger = \Omega(n)$ from Step~3.
 
\medskip
\noindent\emph{Combining everything.}
By \cite[Cor.~5.3]{SV19}, the final $\gamma_2$-$\mathrm{MDP}$ instance has rank at most $n^\dagger + 1$, block length $m = m_{\mathrm{NCP}} + k_{\mathrm{rep}}\,m^\dagger$, and threshold $t = t_{\mathrm{NCP}} + k_{\mathrm{rep}}\,r^\dagger$. We verify each parameter.
 
\smallskip
\noindent\emph{Rank.} The upper bound is $n^\dagger + 1 = O(n)$ by Step~3. For the lower bound, recall the generator produced by \cite[Thm.~5.2]{SV19} (applied by \cite[Cor.~5.3]{SV19} with the stacked gadget): its first $n^\dagger$ columns generate the subcode
\[
   W \;=\; \big\{ \big(CTz^\dagger,\; \underbrace{C^\dagger z^\dagger, \dots, C^\dagger  z^\dagger}_{k_{\mathrm{rep}}}\big) : z^\dagger \in  \F_2^{n^\dagger}\big\}.
\]
Since generator matrices are full column rank (a standing convention of \cite[\S1]{SV19}), the map $z^\dagger \mapsto C^\dagger z^\dagger$ is injective, so $|W| \ge 2^{n^\dagger}$ and the output code has rank at least $n^\dagger$. Hence
\[
   \text{rank} \in \{\,n^\dagger,\; n^\dagger + 1\,\} = \Theta(n).
\]
 
\smallskip
\noindent\emph{Block length.}
\[
   m \;=\; 7\,m_{\mathrm{SAT}} + k_{\mathrm{rep}}\,m^\dagger
     \;=\; \Theta(n) + \Theta(1)\cdot\Theta(n) \;=\; \Theta(n).
\]
 
\smallskip
\noindent\emph{Threshold.}
\[
   t \;=\; t_{\mathrm{NCP}} + k_{\mathrm{rep}}\,r^\dagger
     \;=\; \Theta(n) + \Theta(1)\cdot\Theta(n)
     \;=\; \Theta(n) \;=\; \Theta(m). \qedhere
\]
\end{proof}

\end{document}